\documentclass[journal]{IEEEtran}
\usepackage[caption=false]{subfig}
\usepackage{graphicx}
\usepackage{multirow}
\usepackage{amsmath,amssymb,bm,mathrsfs,xfrac,mathtools,mathrsfs,amsthm}
\usepackage{color}
\usepackage[export]{adjustbox}
\usepackage{float}
\usepackage{cite}
\usepackage[hyphens]{url}
\usepackage[hidelinks]{hyperref}
\usepackage{algorithm}
\usepackage[noend]{algpseudocode}

\newcommand*{\tran}{^{\mkern-1.5mu\mathsf{T}}}

\newtheorem{definition}{Definition}
\newtheorem{remark}{Remark}

\newtheorem{lemma}{Lemma}

\ifCLASSINFOpdf
\else
\fi

\begin{document}
\bstctlcite{IEEEexample:BSTcontrol}

\title{Robust Dynamic Mode Decomposition}

\author{
Amir~Hossein~Abolmasoumi,~Marcos~Netto,~\IEEEmembership{Member,~IEEE,}
and~Lamine~Mili,~\IEEEmembership{Life~Fellow,~IEEE}
\thanks{This work was authored in part by the National Renewable Energy Laboratory (NREL), operated by Alliance for Sustainable Energy, LLC, for the U.S. Department of Energy (DOE) under contract no. DE-AC36-08GO28308. This work was supported by the Laboratory Directed Research and Development (LDRD) program at NREL. The views expressed in the article do not necessarily represent the views of the DOE or the U.S. Government. The U.S. Government and the publisher, by accepting the article for publication, acknowledges that the U.S. Government retains a nonexclusive, paid-up, irrevocable, worldwide license to publish or reproduce the published form of this work, or allow others to do so, for U.S. Government purposes.}
\thanks{Amir H. Abolmasoumi is with the Electrical Engineering Department, Arak University, Arak, 38156879, Iran, and was a visiting scholar with the Bradley Department of Electrical and Computer Engineering, Virginia Polytechnic Institute and State University, Falls Church, VA 22043, USA. M. Netto is with the Power Systems Engineering Center, NREL, Golden, CO 80401, USA. L. Mili is with the Bradley Department of Electrical and Computer Engineering, Virginia Polytechnic Institute and State University, Falls Church, VA 22043, USA. Corresponding author: \href{mailto:a-abolmasoumi@araku.ac.ir}{a-abolmasoumi@araku.ac.ir}.}
}

% \markboth{IEEE TRANSACTIONS ON SIGNAL PROCESSING,~Vol.~, No.~, ~2021}{Shell \MakeLowercase{\textit{et al.}}: Bare Demo of IEEEtran.cls for Journals}

\maketitle

{\color{black}
\begin{abstract}
This paper develops a robust dynamic mode decomposition (RDMD) method endowed with statistical and numerical robustness. Statistical robustness ensures estimation efficiency at the Gaussian and non-Gaussian probability distributions, including heavy-tailed distributions. The proposed RDMD is statistically robust because the outliers in the data set are flagged via projection statistics and suppressed using a Schweppe-type Huber generalized maximum-likelihood estimator that minimizes a convex Huber cost function. The latter is solved using the iteratively reweighted least-squares algorithm that is known to exhibit a better convergence property and numerical stability than the Newton algorithms. Several numerical simulations using canonical models of dynamical systems demonstrate the excellent performance of the proposed RDMD method. The results reveal that it outperforms several other methods proposed in the literature.
\end{abstract}
}

\begin{IEEEkeywords}
Dynamic mode decomposition; Outlier detection; Robust estimation; Robust statistics; Robust regression.
\end{IEEEkeywords}

\IEEEpeerreviewmaketitle

\section{Introduction}
{\color{black}
\IEEEPARstart{T}{he} sustained growth of data acquisition across all areas of human activity is a crucial driver for the research and development of data science methods \cite{Hey2009, Brunton2019}. This fact especially applies to complex dynamical systems for which first-principles models are challenging to obtain while a large amount of data are available. 

A wealth of data science methods have been developed by researchers and made available to practitioners. Dynamic mode decomposition (DMD) stands out because of its connection with the Koopman operator theory \cite{Rowley2009}, which reconciles data analysis and the mathematical knowledge of dynamical systems; the reader is referred to \cite{Kutz2016, Mauroy2020} for more details. Since the publication of the paper authored by Schmid and Sesterhenn \cite{Schmid2008, Schmid2010}, DMD has become the mainstream method for data-driven modeling of dynamical systems, mainly applied to fluid mechanics \cite{Rowley2009}, electric power grids \cite{Susuki2011}, neuroscience \cite{Bingni2016}, finance \cite{Mann2016}, climate science \cite{Kutz2016b}, and transportation \cite{Avila2020}, to name a few.

The original DMD \cite{Schmid2008, Schmid2010} and most of its variants that are tailored to specific classes of dynamical systems make use of a least-squares estimator. Examples of these variants include but are not limited to multiresolution DMD \cite{Kutz2016b}, DMD with control \cite{Proctor2016}, Hankel DMD \cite{Brunton2017}, and tensor-based DMD \cite{Klus2018}. Section \ref{sec.DMD} briefly introduces the original DMD and makes explicit the least-squares estimator. The latter is of particular interest in this paper, as discussed next.

In the classic literature in robust statistics \cite{Huber1964, Wainer1976}, one defines \emph{robustness} as insensitivity to deviations from the assumptions. In this sense, the least-squares estimator is not robust. Two cases of deviations from the assumptions are of particular concern. The first case arises when the probability distribution of the observations is not Gaussian. The least-squares estimator quickly loses its statistical efficiency (that is, accuracy) when the tails of the probability distribution of the observations become slightly thicker than the Gaussian distribution or when the probability distribution of the observations becomes slightly asymmetric. The second case arises when the probability distribution of the majority of the observations is Gaussian except for a few observations, which may take arbitrary values. In this respect, one defines an \emph{outlier} as a data point that violates the underlying assumptions---in other words, it is a data point that is distant from the majority of the point cloud \cite{Rousseeuw2005}. The least-squares estimator produces strongly biased results in the presence of a single outlier in the data set \cite{Maronna2019}. Both cases of deviations from the assumptions often occur in practice---for example, when the probability distribution of the observations is not known while being assumed to be Gaussian or when outliers arise because of instrumentation and communications errors or a poor experimental setup. This fact precludes the DMD from being applied to practical settings, especially for control purposes where pre-cleaning the data set is not an option. 

It turns out that solving the sensitivity of DMD to deviations from the assumptions made about the data set is a challenging task \cite{Wu2021} due to the vulnerability of the least-squares estimator to non-Gaussian noise and outliers, which is a great concern to practitioners. This fact motivated several independent investigations to assess the accuracy of the DMD in capturing the underlying system dynamics directly from the data set \cite{Duke2012, Azencot2019, Zhang2020, Lu2020}. For instance, Dawson $et$  $al.$  \cite{Dawson2016} and Hemati $et$ $al.$ \cite{Hemati2017} address, respectively, the bias introduced by Gaussian noise and the bias resultant from asymmetrically processing snapshots. In Section \ref{sec.Results} of this paper, numerical experiments confirm that the DMD variant proposed by Hemati $et$ $al.$ \cite{Hemati2017} has excellent performance in the presence of Gaussian noise and has good performance in the presence of non-Gaussian but symmetrically distributed noise. This is achieved thanks to a reformulation of the DMD using a total least-squares estimator \cite{Markovsky2007}; however, this estimator is still vulnerable to outliers.

The vulnerability of the least-squares estimator to outliers is not directly solvable without data preprocessing; therefore, Askham \emph{et al.} \cite{Askham2017} reformulate DMD as an optimization problem and make use of a least trimmed squares (LTS) estimator, specifically the trimmed M-estimator introduced by Rousseeuw \cite{rousseeuw1985multivariate}. To the best of the authors' knowledge, \cite{Askham2017} is the only formulation of DMD that makes use of a robust estimator. In particular, the LTS estimator has a high breakdown point \cite{rousseeuw1985multivariate}---that is, this estimator is very robust from a statistical standpoint; however, the formulation in \cite{Askham2017} lacks a mechanism to identify outliers. Indeed, identifying outliers without access to a system model is challenging but necessary in DMD. The formulation in \cite{Askham2017} circumvents this challenge by making a blanket assumption that the time-series data can be represented ``by the outer product of a matrix of exponentials, representing Fourier-like time dynamics, and a matrix of coefficients, representing spatial structures.'' Consequently, nonexponential dynamics in the data set are, therefore, classified as outliers. This fact precludes the application of the method proposed in \cite{Askham2017} to dynamical systems that present nonexponential dynamics.

To this point, the discussion centers around the limitation of the DMD concerning statistical robustness; however, numerical robustness, also referred to as numerical sensitivity \cite{Chen2012}, is equally important. Numerical robustness has severe implications for the stability and convergence of numerical methods. Compared to the original DMD method \cite{Schmid2008, Schmid2010}, the DMD reformulation as an optimization problem proposed by Chen $et$ $al.$  \cite{Chen2012} is less numerically sensitive to deviations from the assumptions. Other researchers have also exploited this approach \cite{Sinha2020}. From a statistical standpoint, although these methods \cite{Chen2012, Sinha2020} provide superior numerical performance than the original DMD method, they are not robust to outliers; therefore, an alternative DMD method that is robust---from both the statistical and numerical standpoints---and is generally applicable to dynamical systems is of great interest to practitioners.

This paper develops an efficient numerical algorithm that makes DMD robust to outliers, even in a position of leverage. Leverage points are measurements whose projections on the factor space are outliers \cite{Rousseeuw2005, Mili1996}. Therefore, the proposed robust dynamic mode decomposition (RDMD) method extends the Schweppe-type Huber generalized maximum-likelihood estimator \cite{Gandhi2010} to \emph{matrix regression problems}. The RDMD method minimizes a convex Huber loss function that incorporates weights calculated via projection statistics. Thus, it can bound the influence of the outliers while maintaining good statistical efficiency at the Gaussian and thick-tailed distributions. Simulations revealed that the RDMD exhibits excellent performance on a collection of canonical models of dynamical systems, including the Van der Pol oscillator and a family of slow-manifold nonlinear systems, under various cases of deviation from the Gaussian assumption. Furthermore, it demonstrates high statistical efficiency under Gaussian and non-Gaussian probability distributions in addition to robustness to outliers. Finally, a numerical comparison against several other methods \cite{Hemati2017, Sinha2020, Askham2017} showcases the performance of the proposed RDMD method.

The paper proceeds as follows. Section II briefly establishes the connection between the Koopman operator and DMD; two numerical procedures for DMD \cite{Schmid2010, Tu2014} are outlined. Section III develops the proposed RDMD method, which is the main contribution of this paper. Section IV discusses the numerical results, and Section V concludes the paper.
}

\section{Preliminaries}

\subsection{Koopman Operator}

Consider an autonomous dynamical system evolving on a finite, $n$-dimensional manifold $\mathbb{X}$ given by

\begin{equation} \label{eq.1}
\bm{x}[k] = \bm{F}(\bm{x}[k-1]), \; \text{for discrete-time } k\in\mathbb{Z},
\end{equation}

\noindent
where $\bm{x}\in\mathbb{X}$ is the state, and $\bm{F}:\mathbb{X} \to \mathbb{X}$ is a nonlinear vector-valued map. Next, we introduce the Koopman operator for discrete-time dynamical systems.

Let $g(\bm{x})$ be a scalar-valued function defined in $\mathbb{X}$, such that $g:\mathbb{X}\to\mathbb{R}$. The function $g$ is referred to as the \emph{observable function}. Let the space of observable functions be $\mathcal{F}\subseteq{C}^{0}$, where ${C}^{0}$ denotes all continuous functions \cite{Mauroy2016}. The Koopman operator, $\mathcal{K}$, is a linear, infinite-dimensional operator \cite{Koopman1931} that acts on $g$ as follows:

\begin{equation}\label{eq.2x}
\mathcal{K}\,g := g \circ \bm{F},
\end{equation}

\noindent
where $\circ$ denotes the function composition. Formally, we have

\begin{equation}
\mathcal{K}g\big(\bm{x}[k]\big) = g\big(\bm{F}(\bm{x}[k-1])\big).
\end{equation}

{\color{black}The interpretation of (\ref{eq.2x}) is as follows. Instead of focusing on the evolution of the state, $\bm{x}$, one shifts the focus to the observables, $\bm{g}(\bm{x})$. The advantage is that the observables evolve linearly with time without neglecting the nonlinear dynamics of the underlying dynamical system given by (\ref{eq.1}).}

\subsection{Dynamic Mode Decomposition}\label{sec.DMD}

For simplicity of notation, define

\begin{equation} \label{eq.4}
\bm{y}_{k} := \bm{g}\big(\bm{x}[k]\big),
\end{equation}

\noindent
where $\bm{y}_{k}\in\mathbb{R}^{m}$ is a vector of $m$ measurements on (\ref{eq.1}) at time $k$. {\color{black}Note that, in principle, any measurement is a function of the state, $\bm{x}$.} In some applications, the measurement set is the state itself, such that $\bm{y}_{k}=\bm{x}_{k}$, and $m=n$. In this paper, we consider the more general case defined in (\ref{eq.4}), where $m\ne n$.

Suppose that one collects sampled measurements of (\ref{eq.1}) at time instances $k=\{0,1,...,N\}$. Define the data matrices as follows:

\begin{equation}\label{eq.5x}
\bm{Y}:=[\bm{y}_{0}\;\;\; \bm{y}_{1}\;\;\; ...\;\;\; \bm{y}_{N-1}], \quad \bm{Y}^{\prime}:=[\bm{y}_{1}\;\;\; \bm{y}_{2}\;\;\; ...\;\;\; \bm{y}_{N}],
\end{equation}

\noindent
where $\bm{Y},\,\bm{Y}^{\prime}\in\mathbb{R}^{m\times{N}}$. For a sufficiently large $N$, one gets \cite{Schmid2010}

\begin{equation} \label{eq.3x}
\bm{Y}^{\prime} \approx \bm{A}\cdot\bm{Y},
\end{equation}

\noindent
{\color{black}where $\bm{A}\in\mathbb{R}^{m\times{m}}$. Note that the operator $\bm{A}$ in (\ref{eq.3x}) pushes the measurement set one step forward in time. For this reason, $\bm{A}$ is a finite-dimensional approximation to the Koopman operator, $\mathcal{K}$. This connection to the Koopman operator gives the DMD method a theoretical support based on the mathematical knowledge of the dynamical systems. Let us now focus on the DMD method.} 

In the original derivation of the DMD \cite{Schmid2008, Schmid2010}, $\bm{A}$ takes the form of a companion matrix. As discussed in \cite{Schmid2010}, however, a practical implementation based on the companion matrix yields an ill-conditioned algorithm. Instead, the following procedure is suggested in \cite{Schmid2010}:

\begin{enumerate}
%1
\item First, compute a reduced singular value decomposition of $\bm{Y}$ as follows:
\begin{equation}
\bm{Y} = \bm{U}\bm{\Sigma}\bm{V}, 
\end{equation}
{\color{black}where $\bm{U}\in\mathbb{C}^{m\times{c}}$, $\bm{\Sigma}\in\mathbb{C}^{c\times{c}}$, $\bm{V}\in\mathbb{C}^{N\times{c}}$, and $c$ is the rank of $\bm{Y}$.}
%2
\item Then, compute
\begin{equation}\label{eq.8x}
\bm{\widetilde{A}} = \bm{U}\bm{Y}\bm{V}\bm{\Sigma}^{-1}.
\end{equation}
\end{enumerate}

\noindent
The numerical procedure goes on with an eigendecomposition of $\bm{\widetilde{A}}$ but, for this paper, the outlined steps suffice. As explained in \cite{Schmid2010}, (\ref{eq.8x}) amounts to a projection of the linear operator $\bm{A}$ onto a proper orthogonal decomposition basis. The method proposed in \cite{Schmid2010} is known as the standard DMD method.

In addition to the standard DMD method, the exact DMD proposed by Tu \emph{et al.} \cite{Tu2014} is widely used. It consists of the following steps:

\begin{enumerate}
%1
\item Let
\begin{equation} \label{eq.7}
\bm{Y}^{\prime} = \bm{A}\cdot\bm{Y}.
\end{equation}
%2
\item Post-multiply both sides of (\ref{eq.7}) by $\bm{Y}\tran$ to get
\begin{equation}
\bm{Y}^{\prime}\bm{Y}\tran = \bm{A}\bm{Y}\bm{Y}\tran.
\end{equation}
%3
\item Then, obtain an estimate of $\bm{A}$ as follows:
\begin{equation} \label{eq.9}
\bm{\widehat{A}} = \bm{Y}^{\prime}\bm{Y}\tran\left(\bm{Y}\bm{Y}\tran\right)^{-1}.
\end{equation}
\end{enumerate}

\noindent Note that $\bm{Y}\tran\left(\bm{Y}\bm{Y}\tran\right)^{-1}$ in (\ref{eq.9}) is the Moore-Penrose inverse $\bm{Y}^{\dagger}$ of a matrix $\bm{Y}$, and the matrix $\left(\bm{Y}\bm{Y}\tran\right)$ is invertible as long as $\bm{Y}$ has linearly independent rows. Given an estimate $\bm{\widehat{A}}$, it is straightforward to compute the approximations to the Koopman eigenvalues and Koopman modes \cite{Rowley2009}.

Now, let the residual (column) vector at time $k$ be defined as

\begin{equation}\label{eq.10x}
\bm{r}_{k} := \bm{y}_{k+1} - \bm{A}\bm{y}_{k},
\end{equation} 

\noindent
and let the $i$th element of $\bm{r}_{k}$ be given by

\begin{equation}
r_{k}^{[i]} = y_{k+1}^{[i]} - \bm{\mathsf{a}}_{i}\tran\bm{y}_{k} \label{eq.11x},
\end{equation}

\noindent
where $\bm{\mathsf{a}}_{i}\tran$ denotes the $i$th row of $\bm{A}$.  It can be shown that (\ref{eq.9}) is the solution to a classic linear least-squares regression problem---that is,

\begin{equation}
\text{minimize} \qquad \sum_{k=0}^{N-1}\sum_{i=1}^{m}\rho_{\ell{s}}\left(r_{k}^{[i]}\right) \label{eq.4x},
\end{equation}

\noindent
where the least-squares loss function $\rho_{\ell{s}}\left(r_{k}^{[i]}\right)=\frac{1}{2}\left(r_{k}^{[i]}\right)^{2}$.

\begin{lemma}
The estimate $\bm{\widehat{A}}$ in (\ref{eq.9}) is the least-squares solution to (\ref{eq.3x}).
\end{lemma}
\begin{proof}
Define

\begin{align}
J_{\ell{s}}(\bm{\mathsf{a}}_{i}) :&= \sum_{k=0}^{N-1}\sum_{i=1}^{m} \frac{1}{2}\left(y_{k+1}^{[i]} - \bm{\mathsf{a}}_{i}\tran\bm{y}_{k}\right) \left(y_{k+1}^{[i]} - \bm{\mathsf{a}}_{i}\tran\bm{y}_{k}\right) \\
&= \sum_{k=0}^{N-1}\sum_{i=1}^{m}\frac{1}{2}\left(r_{k}^{[i]}\right)^{2} = \sum_{k=0}^{N-1}\sum_{i=1}^{m}\rho_{\ell{s}}\left(r_{k}^{[i]}\right) \nonumber.
\end{align}
To minimize $J_{\ell{s}}(\bm{\mathsf{a}}_{i})$, one takes its partial derivative with respect to $\bm{\mathsf{a}}_{i}$ and sets it equal to zero. Formally, we have

\begin{align}
\frac{\partial{J_{\ell{s}}}(\bm{\mathsf{a}}_{i})}{\partial\bm{\mathsf{a}}_{i}} &= \sum_{k=0}^{N-1}\sum_{i=1}^{m} \frac{\partial\rho_{\ell{s}}\left(r_{k}^{[i]}\right)}{\partial{r}_{k}^{[i]}} \cdot \frac{\partial{r}_{k}^{[i]}}{\partial\bm{\mathsf{a}}_{i}} \nonumber \\
&= -\sum_{k=0}^{N-1}\sum_{i=1}^{m} \psi_{\ell{s}}\left(r_{k}^{[i]}\right)\bm{y}_{k} = \bm{0}_{m} \label{eq.14x},
\end{align}

\noindent
where $\bm{0}_{m}$ denotes a column vector of dimension $m$, which has all elements equal to zero, and

\begin{equation} \label{eq.15x}
\psi_{\ell{s}}\left(r_{k}^{[i]}\right) = r_{k}^{[i]}
\end{equation}

\noindent
is the least-squares $\psi-$function, also known as the score function. From (\ref{eq.14x}), one has

\begin{align}
\sum_{k=0}^{N-1} \bm{r}_{k}\bm{y}_{k}\tran &= \sum_{k=0}^{N-1} \left(\bm{y}_{k+1} - \bm{A}\bm{y}_{k}\right)\bm{y}_{k}\tran
= \sum_{k=0}^{N-1} \bm{y}_{k+1}\bm{y}_{k}\tran - \bm{A}\bm{y}_{k}\bm{y}_{k}\tran \nonumber \\
&= \bm{Y}^{\prime}\bm{Y}\tran - \bm{A}\bm{Y}\bm{Y}\tran = \bm{0}_{mm} \label{eq.16x},
\end{align}

\noindent
where $\bm{0}_{mm}$ denotes a square matrix of dimension $m$, which has all elements equal to zero. Finally, from (\ref{eq.16x}), we have

\begin{equation}
\bm{A} = \bm{Y}^{\prime}\bm{Y}\tran\left(\bm{Y}\bm{Y}\tran\right)^{-1},
\end{equation}

\noindent
and the proof is complete.
\end{proof}

\begin{remark}
In statistics, the score function is the gradient of the log-likelihood function with respect to the parameter vector. Evaluated at a particular point of the parameter vector, the score indicates the steepness of the log-likelihood function and thereby the sensitivity to infinitesimal changes to the parameter values. If the log-likelihood function is continuous over the parameter space, then the score will vanish at a local maximum or minimum; this fact is used in the maximum-likelihood estimation to find the parameter values that maximize the likelihood function.
\end{remark}

\begin{remark}
A bounded influence function, which is proportional to the score function, is a necessary condition for a robust estimator. From (\ref{eq.15x}), it is clear that the least-squares score function is unbounded, and, therefore, the least-squares  estimator is not robust.
\end{remark}

Hence, it is expected that a least-squares estimator provides strongly biased results when the samples contained in the data matrices $\bm{Y}$ and $\bm{Y}^{\prime}$ are contaminated with outliers. As a result, in the presence of outliers, the exact DMD has an unbounded bias. Moreover, it can be shown that singular value decomposition is also based on a least-squares estimator; therefore, in the presence of outliers, the standard DMD method has also an unbounded bias.

\section{Robust Exact Dynamic Mode Decomposition}
It is necessary to modify the exact DMD method to construct a statistically robust DMD method. First, outliers must be detected and identified. To this end, we rely on projection statistics to derive weights over the interval $[0, 1]$ that are used to bound the influence of outliers. Specifically, the farther an outlier is from the center of the data cloud, the smaller its assigned weight. All remaining data points not identified as outliers receive a weight equal to $1$. These weights are incorporated into the Huber loss function to bound the influence of outliers in the estimation process. The details are presented next. 

\subsection{Multidimensional Outlier Detection}
The detection and identification of outliers are key steps in robust statistics. In statistical analysis, several methods have been proposed to detect an outlier based on its distance from the majority of the data point cloud, as explained next. 

Let a univariate data set, $\mathbb{P}\subseteq\mathbb{R}$, be $\{p_{1},...,p_{N}\}$. A measure of the distance between a data point, $p_{k}\in\mathbb{P}$, and the center of the data cloud is given by $\frac{p_{k}-\widehat{\ell}}{\widehat{s}}$, where $\widehat{\ell}$ denotes an estimator of location, and $\widehat{s}$ denotes an estimator of scale. A classic measure of distance in the univariate case is provided by 

\begin{equation} \label{eq.17x}
d(p_{k}) = \frac{p_{k}-\widehat{\mu}_{\mathbb{P}}}{\widehat{\sigma}_{\mathbb{P}}},
\end{equation}

\noindent
where the sample mean of the data points in ${\mathbb{P}}$, $\widehat{\mu}_{\mathbb{P}}$, is used as an estimator of location and the sample standard deviation of the data points in ${\mathbb{P}}$, $\widehat{\sigma}_{\mathbb{P}}$, is used as an estimator of scale. Note that $d(p_{k})$ is often referred to as the z-score of the data point $p_{k}$. The classic measure of distance given by (\ref{eq.17x}) is generalized to the multivariate case by the Mahalanobis distance.

\begin{definition}[Mahalanobis distance]
Let a multivariate data set, $\mathbb{P}\subseteq\mathbb{R}^{m}$, be $\{\bm{p}_{1},...,\bm{p}_{N}\}$. The Mahalanobis distance between a data point, $\bm{p}_{k}$, and the data cloud comprising all data points in $\mathbb{P}$ is defined as
\begin{equation}
d_{M}(\bm{p}_{k}) := \left[\left(\bm{p}_{k}-\widehat{\bm{\mu}}\right)\tran \widehat{\bm{S}}^{-1} \left(\bm{p}_{k}-\widehat{\bm{\mu}}\right)\right]^{1/2} \label{eq.18x},
\end{equation}

\noindent
where $\widehat{\bm{\mu}}=\left(\widehat{\mu}_{1},...,\widehat{\mu}_{N}\right)\tran$ and $\widehat{\bm{S}}$ are, respectively, the sample mean and the sample covariance matrix of the data points in $\mathbb{P}$.
\end{definition}
By comparing (\ref{eq.17x}) and (\ref{eq.18x}), note that the sample standard deviation used in the univariate case is replaced by the sample covariance matrix in the multivariate case.

It can be shown (see \cite{Rousseeuw2005}) that for a scalar $b\in\mathbb{R}$, the set of data points for which $d_{M}^{2}<b$ lies inside an ellipsoid with the center at $\widehat{\bm{\mu}}$. Moreover, if the data points in $\mathbb{P}\subseteq\mathbb{R}^{m}$ follow a multivariate normal distribution, then the values of $d_{M}^{2}$ follow a chi-square distribution with $m$ degrees of freedom, $\chi_{m}^{2}$; hence, there is a probability of $1-\alpha$ that a data point $\bm{p}_{k}$ such that $d_{M}^{2}\le\chi_{m,1-\alpha}^{2}$ is located within an ellipsoid given by $d_{M}^{2}=\chi_{m,1-\alpha}^{2}$ that is centered at $\widehat{\bm{\mu}}$. This provides the rationale used to tag outliers---that is, an outlier is any data point for which the Mahalanobis distance is larger than a threshold, e.g., $\left(\chi_{m,0.975}^{2}\right)^{1/2}$. But because $d_{M}$ is calculated via non-robust estimators of location and scale, it is vulnerable to the masking effect of multiple outliers, especially when the latter appear in clusters \cite{Rousseeuw1990}. In other words, the corresponding ellipsoid is inflated to the point that it encompasses outliers, which can no longer be identified. To gain robustness, one can replace the sample mean and the sample covariance matrix in (\ref{eq.18x}) by robust estimators of location and scale, respectively. This is discussed next.

\begin{definition}[Median absolute deviation from median in the case of a univariate data set]
Let a univariate data set, $\mathbb{P}\subseteq\mathbb{R}$, be $\{p_{1},...,p_{N}\}$. A very robust estimator of scale \cite{Rousseeuw1993} is the median absolute deviation from the median, which is defined as
\begin{equation}
\textnormal{mad}_{\mathbb{P}} := 1.4826 \cdot \textnormal{median} \left|p_{k} - \textnormal{median} \left( \bm{p}\tran \right) \right|
\end{equation}

\noindent
for $k=\{1,...,N\}$, where $\bm{p}\tran=[p_{1}\;p_{2}\;...\;p_{N}]$, and the constant 1.4826 makes the estimator consistent at normal distributions.
\end{definition}

In the univariate case, a robust distance between $p_{k}$ and the center of the data cloud is given by:

\begin{equation}
d_{r}(p_{k}) = \frac{\left|p_{k}-\text{median}\left(\bm{p}\tran\right)\right|}{\text{mad}_{\mathbb{P}}}.
\end{equation}

\begin{definition}[Median absolute deviation from the median in the case of a multivariate data set]
Let a multivariate data set, $\mathbb{P}\subseteq\mathbb{R}^{m}$, be $\{\bm{p}_{1},...,\bm{p}_{N}\}$. The median absolute deviation from the median is defined as
\begin{equation}
\widehat{s}_{1} := 1.4826 \cdot \textnormal{median}_{k} \big(\left| \bm{p}_{k}\tran\bm{v} - \textnormal{median}_{j} \left( \bm{p}_{j}\tran\bm{v} \right) \right|\big), \label{eq.14}
\end{equation}

\noindent
for $k,\,j=\{1,2,...,N\}$, where $\bm{v}$ is the direction to which the data points are projected. 
\end{definition}

In the multivariate case, however, it is challenging to align all multivariate data points such that a meaningful measure of distance can be obtained. A solution to this challenge stems from the fact that the Mahalanobis distance can be written as follows \cite{Donoho1982}:

\begin{align}
d_{M}(\bm{p}_{k}) &= \left[\left(\bm{p}_{k}-\widehat{\bm{\mu}}\right)\tran \widehat{\bm{S}}^{-1} \left(\bm{p}_{k}-\widehat{\bm{\mu}}\right)\right]^{1/2} \nonumber \\
&\equiv \max_{\left|\bm{v}\right|=1} \left(\frac{\left| \bm{p}_{k}\tran\bm{v} - \widehat{\ell} \left( \bm{p}_{1}\tran\bm{v}, ..., \bm{p}_{N}\tran\bm{v} \right) \right|}{\widehat{s}\left( \bm{p}_{1}\tran\bm{v}, ..., \bm{p}_{N}\tran\bm{v} \right)}\right), \label{eq.22x}
\end{align}

\noindent
where $\widehat{\ell}$ and $\widehat{s}$ denote, respectively, an estimator of location and scale. Note that the maximization should be considered on all possible directions $\bm{v}$. To robustify (\ref{eq.22x}), Donoho and Gasko \cite{Donoho1992} suggest using the sample median as the estimator of location and the median absolute deviation from the median given by (24), $\widehat{s}_{1}$, as the estimator of scale. This distance is referred to as the \emph{projection statistic}, and it is defined as

\begin{equation}
d_{\text{ps}}(\bm{p}_{k}) = d_{ps,\,k} := \max_{\left|\bm{v}\right|=1} \left(\frac{\left| \bm{p}_{k}\tran\bm{v} - \text{median}_{j} \left( \bm{p}_{j}\tran\bm{v} \right) \right|}{\widehat{s}_{1}}\right), \label{eq.17}
\end{equation}

\noindent
where the sample median and $\widehat{s}_{1}$ are calculated on the direction of all feasible unit vectors $\bm{v}$. Unfortunately, $\widehat{s}_{1}$ loses statistical efficiency for asymmetric distributions. To address this issue, Croux and Rousseeuw \cite{Croux1992, Rousseeuw1993} propose another robust estimator of scale, which is statistically efficient for asymmetric distributions, and it is defined as

\begin{align} \label{eq.24x}
\widehat{s}_{2} := 1.1926 \cdot \text{lomed}_{k} \left(\right. \text{lomed}_{j\ne k}\left|\right. &\bm{p}_{k}\tran\bm{v} - \bm{p}_{j}\tran\bm{v} \left.\right| \left.\right),
\end{align}

\noindent
for $k,j=\{1,...,N\}$, where lomed denotes a low median---that is, the $\big((N+1)/2\big)$-th order statistic of $N$ data points. Eq. (\ref{eq.24x}) reads as follows: for each $k$, we compute the low median of $|\bm{p}_{k}-\bm{p}_{j}|$ for $j=\{1,...,N\}$. This yields $N$ data points, the low median of which gives the final estimate, $\widehat{s}_{2}$. The factor $1.1926$ is for consistency at the normal distribution.

Yet, another challenge encountered in calculating the projection statistics is that considering all possible directions of $\bm{v}$ cannot be realized. To address this, Gasko and Donoho \cite{Gasko1982} suggest considering only the directions that originate from the coordinate-wise median vector, $\bm{v}_{\text{med}}$, and that pass through each data point $\bm{p}_{k}$, yielding

\begin{equation}
\bm{v}_{k} = \bm{p}_{k} - \bm{v}_{\text{med}},
\end{equation}

\noindent
where

\begin{equation} \label{eq.19}
\bm{v}_{\text{med}} = \left[ \text{median}_{j}(\bm{x}_{j1})\;\; \text{median}_{j}(\bm{x}_{j2})\;\; ...\;\; \text{median}_{j}\bm{x}_{jm} \right]\tran.
\end{equation}

\noindent
Thus, it is enough to investigate $m$ directions. Moreover, it is not necessary to consider $\bm{v}$ to be of unit length because $\widehat{s}_{1}$ and $\widehat{s}_{2}$ are affine equivariant scale estimators.

As is the case for the Mahalanobis distance, the projected statistics approximately follow a chi-square distribution with $m$ degrees of freedom when the data points follow a multivariate normal distribution \cite{Rousseeuw2005}; hence, the projection statistic of each data point, $\bm{p}_{k}$, is calculated, and if they exceed a threshold, e.g., $d_{\text{ps,\,k}}^{2}> \chi_{2,0.975}^{2}$, then the associated data point is tagged as an outlier. Note that in this paper, the data points, $\bm{p}_{k}$, are the columns of the data matrices given by (\ref{eq.5x}).

Next, we develop a mechanism to suppress the adverse effect of the outliers on the estimation process. This is achieved by defining weights, which are calculated via projection statistics and incorporated into the Huber loss function.

\subsection{Generalized Maximum-Likelihood Robust Estimation}

\begin{definition}[Huber loss function]
The Huber loss function is defined as:

\begin{equation}
\rho_{H}\left(r_{k}^{[i]}\right) := 
\begin{cases}
\begin{array}{ll}
\frac{1}{2}\left(r_{k}^{[i]}\right)^{2} & \text{for}\,|r_{k}^{[i]}|\le\delta, \\
\delta|r_{k}^{[i]}| - \frac{1}{2}\delta^{2} & \text{otherwise}.
\end{array}
\end{cases} \label{eq.23}
\end{equation}
\end{definition}
Note that $\rho_{H}(\cdot)$ is quadratic for $|r_{k}^{[i]}|\le\delta$, and linear otherwise. The quadratic and linear sections connect at the point where $|r_{k}^{[i]}|$ is equal to the scalar-valued parameter, $\delta$, which dictates the slope of the function. The parameter $\delta$ is usually adjusted to have a numerical value between $1$ and $3$ to have high statistical efficiency at the normal distribution \cite{Rousseeuw2005}. In this work, we set $\delta=1.5$.

Now, let the {\color{black}Schweppe-type Huber generalized maximum-likelihood} estimator be defined such that it minimizes a convex objective function given by

\begin{equation}
J_{H}\left(\bm{\mathsf{a}}_{i}\right) = \sum_{k=0}^{N-1}\sum_{i=1}^{m} w_{k}^{2}\cdot\rho_{H}\left(r_{ks}^{[i]}\right), \label{eq.21}
\end{equation}

\noindent
where 

\begin{align}
r_{ks}^{[i]} &= \frac{r_{k}^{[i]}}{s\cdot w_{k}} = \frac{1}{s\cdot w_{k}} \left( y_{k+1}^{[i]} - \bm{\mathsf{a}}_{i}\tran\bm{y}_{k} \right), \\
s &= 1.4826 \cdot b_{m} \cdot \text{median}\left|\bm{r}_{k}\tran\right|, \\
w_{k} &= \min\left( 1,\,\frac{b}{d_{ps,\,k}^{2}} \right) \label{eq.31x}.
\end{align}

\noindent
Here, $s$ is a robust estimator of scale, $b_{m}$ is a correction factor, $b$ is set equal to 1.5 for best statistical efficiency and to avoid large biases \cite{Rousseeuw2005}, and $w_{i}$ are weights determined via projection statistics. The calculated weights play an essential role in the development of the RDMD method. {\color{black}Note that $\bm{\mathsf{a}}_{i}$ is the vector-valued variable being optimized, and the cost function is convex with respect to the residuals $r_{k}^{[i]}$.}

{\color{black}We stress that (\ref{eq.21}) is a convex objective function. This fact offers a significant computational advantage over non-convex approaches. Further, note that the Huber loss function (\ref{eq.23}) is quadratic for small values of $r_{k}$, and linear for large values. This characteristic yields a bounded score function, despite the loss function being convex. The latter is an advantage over, e.g., the least-absolute value estimator, which is robust but has a non-convex loss function, or the least-squares estimator, which has a convex loss function but is not robust.}

Having introduced a robust outlier detection algorithm and a generalized maximum-likelihood estimator that can suppress the bias introduced by the outliers, we now present the main result of this paper.

\subsection{Robust Dynamic Mode Decomposition}\label{sec.RDMD}

\noindent
Let $\bm{y}^{\prime}=\text{col}(\bm{Y}^{\prime})=[\bm{y}_{1}\tran\; ...\; \bm{y}_{N}\tran]\tran$, $\bm{a}=\text{row}(\bm{A})=[\bm{\mathsf{a}}_{1}\tran\; ...\; \bm{\mathsf{a}}_{m}\tran]\tran$, and $\bm{B}=\bm{I}_{m}\otimes\bm{Y}\tran$, where $\bm{I}_{m}$ denotes an identity matrix of dimension $m$, and $\otimes$ denotes the Kronecker product. The linear regression (\ref{eq.7}) can be rewritten as:

\begin{align}
\bm{y}^{\prime} = \bm{B}\bm{a} &= \left(\bm{I}_{m}\otimes\bm{Y}\tran\right)\bm{a} \nonumber\\
&=
\left[
\begin{array}{c c c c}
\bm{Y}\tran & & & \\
& \bm{Y}\tran & & \\
& & \ddots & \\
& & & \bm{Y}\tran
\end{array} 
\right]
\left[
\begin{array}{c}
\bm{\mathsf{a}}_{1} \\
\bm{\mathsf{a}}_{2} \\
\vdots \\
\bm{\mathsf{a}}_{m}
\end{array}
\right]. \label{eq.33}
\end{align}

A solution to (\ref{eq.33}) can be found by solving $m$ subproblems of the form given by

\begin{equation}
\bm{y}^{\prime} = \bm{Y}\tran\bm{\mathsf{a}}_{i}, \quad i=\{1,2,...,m\}.
\end{equation}

Hence, for each subproblem $i$, we seek a robust estimate, $\bm{\widehat{\mathsf{a}}}_{i}$, that is the solution to $\bm{y}^{\prime} = \bm{Y}\tran\bm{\mathsf{a}}_{i}$. Note that the residues, $r_{k}^{[i]}$, given by (\ref{eq.11x}) naturally apply to each subproblem $i$. Following (\ref{eq.21}), for each subproblem $i$, we minimize the cost function defined as

\begin{equation}
J_{H}\left(\bm{\mathsf{a}}_{i}\right) = \sum_{k=0}^{N-1} w_{k}^{2}\cdot\rho_{H}\left(r_{ks}^{[i]}\right), \label{eq.35x}
\end{equation}

\noindent
and the optimal solution to (\ref{eq.35x}) satisfies

\begin{equation}
\frac{\partial J_{H}\left(\bm{\mathsf{a}}_{i}\right)}{\bm{\mathsf{a}}_{i}} = \sum_{k=0}^{N-1} -\frac{w_{k}\bm{y}_{k}}{s}\cdot\psi_{H}\left(r_{ks}^{[i]}\right) = \bm{0}_{m}, \label{eq.36x}
\end{equation}

\noindent
where $\psi_{H}\left(r_{ks}^{[i]}\right) = \partial \rho\left(r_{ks}^{[i]}\right)/r_{ks}^{[i]}$, and

\begin{equation}
\psi_{H}\left(r_{ks}^{[i]}\right) = 
\begin{cases}
\begin{array}{ll}
r_{ks}^{[i]} & \text{for}\,\left|r_{k}^{[i]}\right|\le\delta, \\
\delta\cdot\text{sign}\left(r_{ks}^{[i]}\right) & \text{otherwise},
\end{array}
\end{cases}
\end{equation}

\noindent
is the Huber score function. An illustration of the least-squares and the Huber loss and score functions is shown in Fig. \ref{fig.1x}. We stress that having a bounded score function is a necessary condition for an estimator to be robust. Finally, a robust solution to (\ref{eq.33}) is given by solving (\ref{eq.36x}) for $i=\{1,2,...,m\}$, as follows. By multiplying and dividing the Huber score function in (\ref{eq.36x}) by $r_{ks}^{[i]}$, and by defining the scalar weight function as $q\left(r_{ks}^{[i]}\right):=\psi_{H}\left(r_{ks}^{[i]}\right)/r_{ks}^{[i]}$, (\ref{eq.36x}) can be expressed in matrix form, as follows:

\begin{figure}[!t]
\begin{center}
\subfloat[Least-squares loss/score function]{\includegraphics[width=0.24\textwidth]{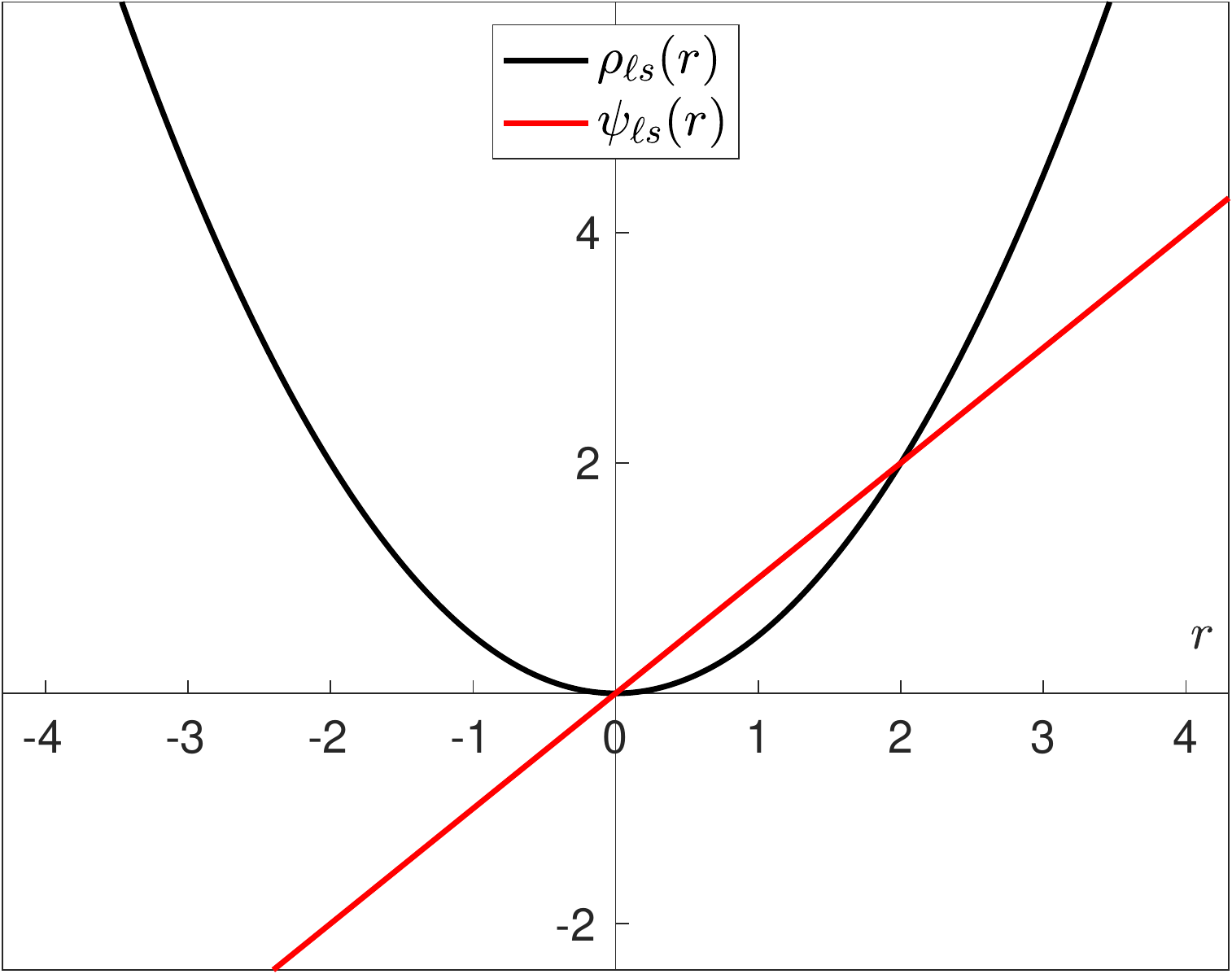}}
\subfloat[Huber loss/score function]{\includegraphics[width=0.24\textwidth]{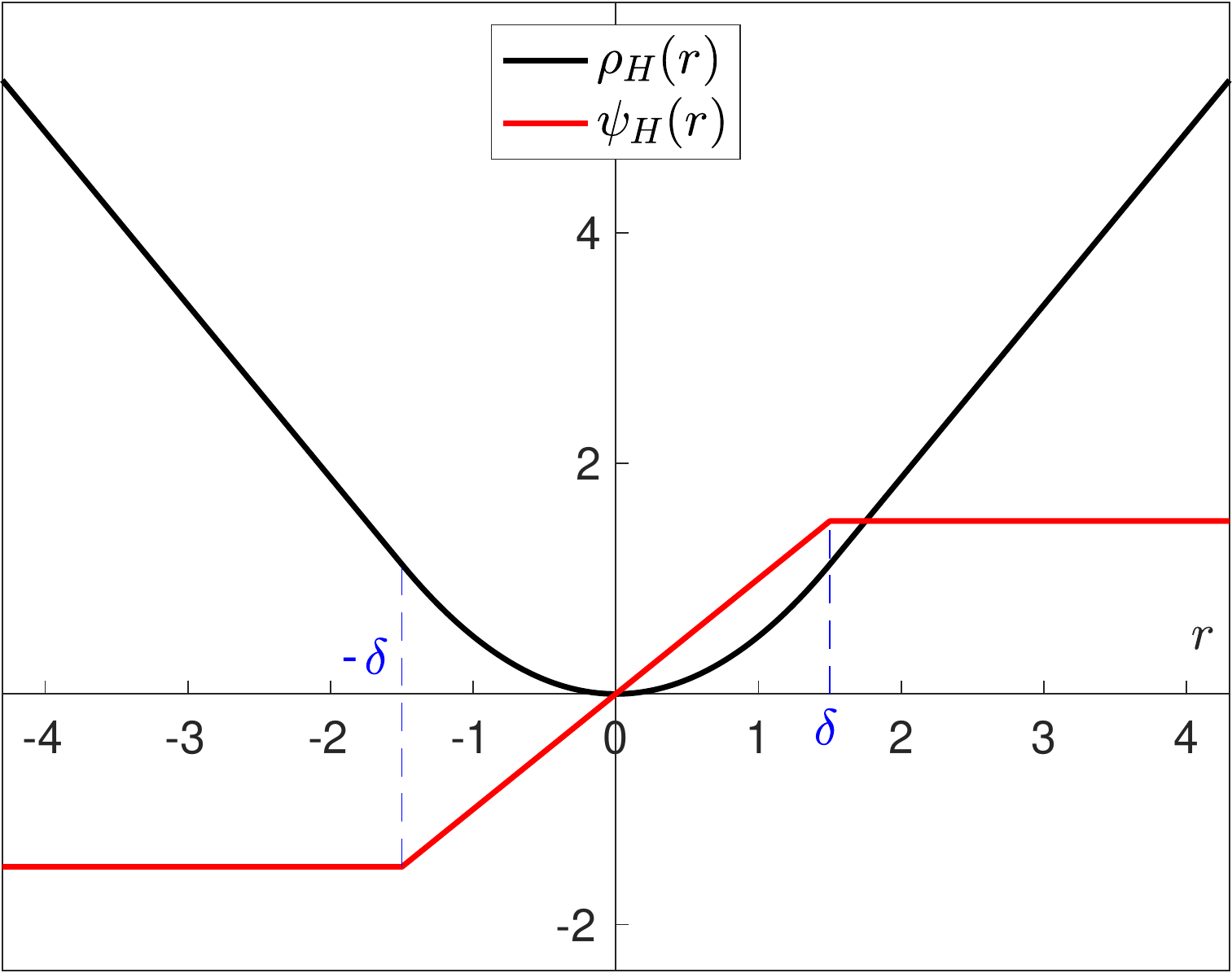}}
\caption{Comparison between (a) least-squares estimator and (b) generalized maximum-likelihood robust estimator based on the Huber loss function.}
\label{fig.1x}
\end{center}
\end{figure}

\begin{equation}
\bm{Y}\bm{Q}\left(\bm{y}^{\prime}-\bm{Y}\tran\bm{\widehat{a}}_{i}\right) = \bm{0},
\end{equation}

\noindent
where $\bm{Q}=\text{diag}\left(q\left(r_{ks}^{[i]}\right)\right)$. Solving for the estimate $\bm{\widehat{a}}_{i}$ using the iteratively reweighted least squares (IRLS) algorithm yields

\begin{equation}
\bm{\widehat{a}}_{i}^{\left(\upsilon+1\right)} = \left(\bm{Y}\bm{Q}^{\left(\upsilon\right)}\bm{Y}\tran\right)^{-1}\bm{Y}\bm{Q}^{\left(\upsilon\right)}\bm{y}^{\prime}, \label{eq.39x}
\end{equation}

\noindent
where the superscript $(\upsilon)$ indicates the $\upsilon$-th iteration. The condition for convergence of the IRLS algorithm is adjusted to meet $||\bm{\widehat{a}}_{i}^{\left(\upsilon+1\right)}-\bm{\widehat{a}}_{i}^{\left(\upsilon\right)}||\le 0.01$. In what follows, the numerical method in (\ref{eq.33})--(\ref{eq.39x}) is referred to as K-RDMD, an allusion to the Kronecker product in (\ref{eq.33}). The computation time of K-RDMD is expected to increase substantially as the number of time instances, $N$, increases. This is addressed next.

Let us redefine the objective function as follows:

\begin{equation}\label{eq.40x}
J_{H}(\bm{A})=\sum_{k=1}^{N-1}w_{k}^{2}\cdot\rho_{H}\left(\bm{r}_{ks}\right), 
\end{equation}

\noindent
where $\rho_{H}\left(\bm{r}_{ks}\right)$ is a modified Huber loss function, as follows:

\begin{equation}
\rho_{H}\left(\bm{r}_{ks}\right) =
\begin{cases}
\begin{array}{ll}
\frac{1}{2}||\bm{r}_{ks}||^{2} & \text{for}\,||\bm{r}_{ks}||\le\delta, \\
\delta\cdot||\bm{r}_{ks}||-\frac{1}{2}\delta^{2} & \text{otherwise},
\end{array}
\end{cases}
\end{equation}

\noindent
where $\bm{r}_{ks}=\bm{r}_{k}/\left(s\cdot{w}_{k}\right)$; $\bm{r}_{k}$ and $\delta$ are as defined before. The optimal solution to (\ref{eq.40x}) is given by

\begin{equation}
\frac{\partial J_{H}(\bm{A})}{\partial\bm{A}} =
\sum_{k=0}^{N-1}
w_{k}^{2}\cdot
\frac{\partial{\rho}_{H}(\bm{r}_{ks})}{\partial||\bm{r}_{ks}||^2}\cdot
\frac{\partial||\bm{r}_{ks}||^2}{\partial||\bm{r}_{k}||^2}\cdot
\frac{\partial||\bm{r}_{k}||^2}{\partial\bm{A}},
\end{equation}

\noindent
where

\begin{align}
\frac{\partial\rho_{H}(||\bm{r}_{ks}||)}{\partial||\bm{r}_{ks}||^2} &= \frac{1}{2||\bm{r}_{ks}||}\frac{\partial{\rho}_{H}(||\bm{r}_{ks}||)}{\partial||\bm{r}_{ks}||} = \frac{\psi_{H}(||\bm{r}_{ks}||)}{2||\bm{r}_{ks}||}, \\
\frac{\partial||\bm{r}_{ks}||^2}{\partial||\bm{r}_{k}||^2} &= \frac{1}{s^{2}w_{k}^{2}}, \\
\frac{\partial||\bm{r}_{k}||^2}{\partial\bm{A}} &= 2\bm{r}_{k}\bm{y}_{k}\tran.
\end{align}

\noindent
Thus, we have

\begin{equation}
\frac{\partial J_{H}(\bm{A})}{\partial\bm{A}} = 
\sum_{k=0}^{N-1} 
w_{k}^{2}
\frac{\psi_{H}(||\bm{r}_{ks}||)}{||\bm{r}_{ks}||}
\frac{\left(\bm{y}_{k+1}-\bm{A}\bm{y}_{k}\right)\bm{y}_{k}\tran}{s^{2}w_{k}^{2}} = \bm{0}. \label{eq.45x}
\end{equation}
By putting (\ref{eq.45x}) in matrix form, we get

\begin{equation}
\bm{Y}\bm{Q}\left(\bm{Y}^{\prime}-\bm{A}\bm{Y}\right)\tran = \bm{0},
\end{equation}

\noindent
where $\bm{Q}=\text{diag}\left({\psi_{H}(||\bm{r}_{ks}||)}/{||\bm{r}_{ks}||}\right)$. 
Solving for the estimate $\bm{\widehat{A}}$ using the IRLS algorithm yields

\begin{equation}\label{eq.48x}
\bm{\widehat{A}}^{(\upsilon+1)} = \bm{Y}^{\prime}\bm{Q}^{(\upsilon)}\bm{Y}\tran\left(\bm{Y}\bm{Q}^{(\upsilon)}\bm{Y}\tran\right)^{-1}.
\end{equation}
The condition for convergence of the IRLS algorithm is set to meet $||\bm{\widehat{A}}^{(\upsilon+1)}-\bm{\widehat{A}}^{(\upsilon)}||_{F}\le 0.01$. In what follows, the numerical method in (\ref{eq.40x})--(\ref{eq.48x}) is referred to as N-RDMD, an allusion to the use of the norm of a residual vector. The performance of the proposed RDMD methods---K-RDMD and N-RDMD---is assessed in Section \ref{sec.Results}.

\subsection{Robust Standard Dynamic Mode Decomposition}
The first steps of the standard DMD method are outlined in Section \ref{sec.DMD}, where $c$ denotes the rank of $\bm{Y}$. Note that it is common to assume $c\le m<N$ for large data sets. Moreover, particularly for model order reduction, one is interested in a projection $\bm{\widetilde{A}}\in\mathbb{R}^{c^{\prime}\times{c}^{\prime}}$, where $c^{\prime}<c$. In this case, one disregards $(c-c^{\prime})$ nonzero elements of $\bm{\Sigma}$ and $(c-c^{\prime})$ columns of $\bm{U}$ and $\bm{V}$. This changes the unitary property of $\bm{U}$ and $\bm{V}$ as follows:

\begin{align}
\bm{U}^{*}\bm{U} &= \bm{I}_{c^{\prime}}, \quad \bm{U}\bm{U}^{*}\ne\bm{I}_{N}, \\
\bm{V}^{*}\bm{V} &= \bm{I}_{c^{\prime}}, \quad \bm{V}\bm{V}^{*}\ne\bm{I}_{m}.
\end{align}

In this section, we discuss how to robustify the standard DMD method, even for the most challenging case described herein. Let $\bm{A}=\bm{T}\bm{\widetilde{A}}\bm{T}^{\dagger}$. Further, let $\bm{T}\in\mathbb{C}^{N\times c^{\prime}}$, such that $\bm{T}^{\dagger}\bm{T}=\bm{I}_c^{\prime}$. Then, the eigenvalues of $\bm{\widetilde{A}}$ are a subset of the eigenvalues of $\bm{A}$; hence, for a known transformation $\bm{T}$, the residues in (\ref{eq.10x}) can be rewritten as

\begin{equation}
\bm{r}_{k} = \bm{y}_{k+1} - \bm{T}\bm{\widetilde{A}}\bm{T}^{\dagger}\bm{y}_{k}.
\end{equation}
Thus, we have:

\begin{equation}
\bm{T}^{\dagger}\bm{Y}\bm{Q}\left(\bm{T}^{\dagger}\bm{Y}^{\prime}-\bm{\widetilde{A}}\bm{T}^{\dagger}\bm{Y}\right)\tran = \bm{0},
\end{equation}

\noindent
yielding to the iterative formula given by

\begin{equation} \label{eq.53x}
\bm{\widetilde{A}}^{(\upsilon+1)} = \bm{T}^{\dagger}\bm{Y}^{\prime}\bm{Q}^{(\upsilon)}\bm{Y}\tran\bm{T}\left(\bm{T}^{\dagger}\bm{Y}\bm{Q}^{(\upsilon)}\bm{Y}\tran\bm{T}\right)^{-1}.
\end{equation}

Note that in (\ref{eq.53x}), the weights calculated via the projection statistics act on the matrix $\bm{Q}$. Also note that if $\bm{T}\bm{T}^{\dagger}=\bm{I}_{N}$, then $\bm{Q}$ is canceled out from (\ref{eq.53x}), leading to an ordinary least-squares solution, $\bm{\Tilde{A}}=\bm{Y}^{\prime}\bm{Y}^{\dagger}$, and we loose robustness; therefore, to guarantee robustness, we should have $\bm{T}\bm{T}^{\dagger}\ne\bm{I}_{N}$. 
To relax such a constraint on (\ref{eq.53x}), we resort to the following iterative procedure:
\begin{equation}\label{eq.54x}
\bm{\widetilde{A}}^{(\upsilon+1)} = \bm{T}^{\dagger}\bm{Y}^{\prime}\bm{Q}^{(\upsilon)}\bm{Y}\tran\bm{T}\left(\bm{T}^{\dagger}\bm{Y}\bm{Q}^{(\upsilon)}\bm{Y}\tran\bm{T} + \gamma^{2}\bm{I}\right)^{-1},
\end{equation} 
This is known as the Tikhonov regularization, which is usually used to solve ill-posed regression problems \cite{Tikhonov1979}. In (\ref{eq.54x}), $\gamma$ is a small positive constant. It also can avoid possible numerical problems in calculating the inverse in (\ref{eq.53x}). 

Although the estimation of $\bm{\widetilde{A}}$ from (\ref{eq.53x}) is robust to outliers, it is not guaranteed that it includes all the dominant modes in the data. In other words, although each eigenvalue of $\bm{\widetilde{A}}$ is also an eigenvalue of the original $\bm{A}$, there is the possibility that the dominant eigenvalues are not included. This is mainly determined by how one selects the reduction matrix, $\bm{T}$. An option is to choose $\bm{T}=\bm{U}$, as is done for the standard DMD method. Note that as previously explained, $\bm{U}$ is not a unitary matrix; therefore, robustness is ensured. On the other hand, because there are outliers in the data matrix, $\bm{Y}$, such a reduction might not be able to cover all dominant modes. It is observed that for a low percentage of outliers among the data points, this selection captures the dominant eigenvalues. Another option is to preprocess the data matrices before performing the singular value decomposition. A comprehensive investigation on the choice of the similarity transformation in RDMD will be addressed in future research.

\section{Numerical Results}\label{sec.Results}
In what follows, the performance of the proposed RDMD method is assessed by using a variety of canonical dynamical systems. Furthermore, the performance of the proposed RDMD method is compared to the performance of

\begin{itemize}
\item the exact DMD method proposed in \cite{Tu2014};
\item the total dynamic mode decomposition (TDMD) method proposed in \cite{Hemati2017};
\item the optimization-based dynamic mode decomposition (ODMD) method proposed in \cite{Sinha2020};
\item the robust trimmed DMD proposed in \cite{Askham2017}.
\end{itemize}.
We start by comparing the performance of the two variations of the RDMD method {\color{black}\footnote{\color{black}The MATLAB code is available at \url{https://github.com/amasoumi60/Robust-Dynamic-Mode-Decomposition}.}} developed in Section \ref{sec.RDMD}.

\subsection{Comparison of K-RDMD and N-RDMD}
Consider a network of $s$ oscillators connected in a ring topology. The differential equation describing the angular displacement of the $k$-th oscillator is given by

\begin{equation}
\ddot{\theta}_{k} + \bm{\ell}_{k}\tran\bm{\theta} +d_{k}\theta_{k} = 0,
\end{equation} 

\noindent
where $\theta_{k}$ and $d_{k}$ are, respectively, the angle and the damping coefficient of the $k$-th oscillator; $\bm{\ell}_{k}\tran$ denotes the $k$-th row of the Laplacian matrix $\bm{\mathcal{L}}$; and $\bm{\theta}=\left[\theta_{1}\;\theta_{2}\;...\;\theta_{s}\right]\tran$. Note that the number of states is $m=2\cdot{s}$, and a ring topology with $s=15$ oscillators is considered. The state-space equations are written as

\begin{align}
\bm{\dot{\theta}} &= \bm{\omega}, \nonumber \\ 
\bm{\dot{\omega}} &= -\bm{\mathcal{L}}\bm{\theta} -\bm{D}\bm{\theta}, \label{eq.ring}
\end{align}

\noindent
$\bm{\omega}=\left[\omega_{1}\;...\;\omega_{s}\right]\tran$, and $\bm{D}=\text{diag}(d_{1},...,d_{s})$. The data are collected with a sample time of $0.01$ second. The damping, $d_{k}$, is set to $0.05$ for all $k$.

Fig. \ref{fig.2x} shows the performance of N-RDMD and K-RDMD for the network of the coupled oscillators with 15 oscillators and 500 snapshots. We observe that both methods are robust to outliers. The reconstruction cumulative error that is defined as $\sum_{i=0}^{k}||{\bm{x}_{i,\text{reconstructed}}-\bm{x}_{i,\text{true}}}||$ is also depicted in Fig. \ref{fig.2x}. A difference in performance is observed in the reconstruction of a state variable, i.e., N-RDMD performs better than K-RDMD in reconstructing the state value. This difference is further investigated under three different scenarios, which are summarized in Table \ref{tab.1}. One can expect that despite the high computational effort by K-RDMD, it performs better than N-RDMD; however, note that K-RDMD solves $m$ separate robust estimation problems, for each row of the matrix $\bm{A}$, using the IRLS algorithm. Note that each IRLS loop has an error threshold that overall leads to a larger error in estimating the matrix $\bm{A}$, whereas N-RDMD uses a single IRLS loop using the Frobenius norm of the error matrix. 

In what follows, the performance of the N-RDMD is compared to other DMD methods available in the literature; in this paper, it is referred to as RDMD. We choose N-RDMD because it has a higher computational efficiency than the K-RDMD. We start by presenting results on a simple linear dynamical system and gradually move toward more complex nonlinear dynamical systems.

\begin{table}[!t]
\centering \scriptsize
\setlength{\tabcolsep}{1.0em}
\caption{Computation time (s) and cumulative errors of K-RDMD and N-RDMD} \vspace{-0.2cm}
\begin{tabular}{l l l l l} \hline
Problem dimension & & $s=15,$ & $s=150,$ & $s=150,$ \\
& & $N=500$ & $N=50$ & $N=500$ \\ \hline
Cumulative error & N-RDMD & $0.1178$ & $0.0045$ & $0.1059$ \\
& K-RDMD & $1.1083$ & $10.979$ & $16.447$ \\ \hline
Computation time & N-RDMD & $1.6747$ & $0.1180$ & $6.4297$ \\
& K-RDMD & $1.7398$ & $1.3491$ & $10.011$ \\ \hline
\end{tabular}
\label{tab.1}
\end{table}

\begin{figure}[!t]
\centering
\subfloat[]{\includegraphics[width=0.25\textwidth]{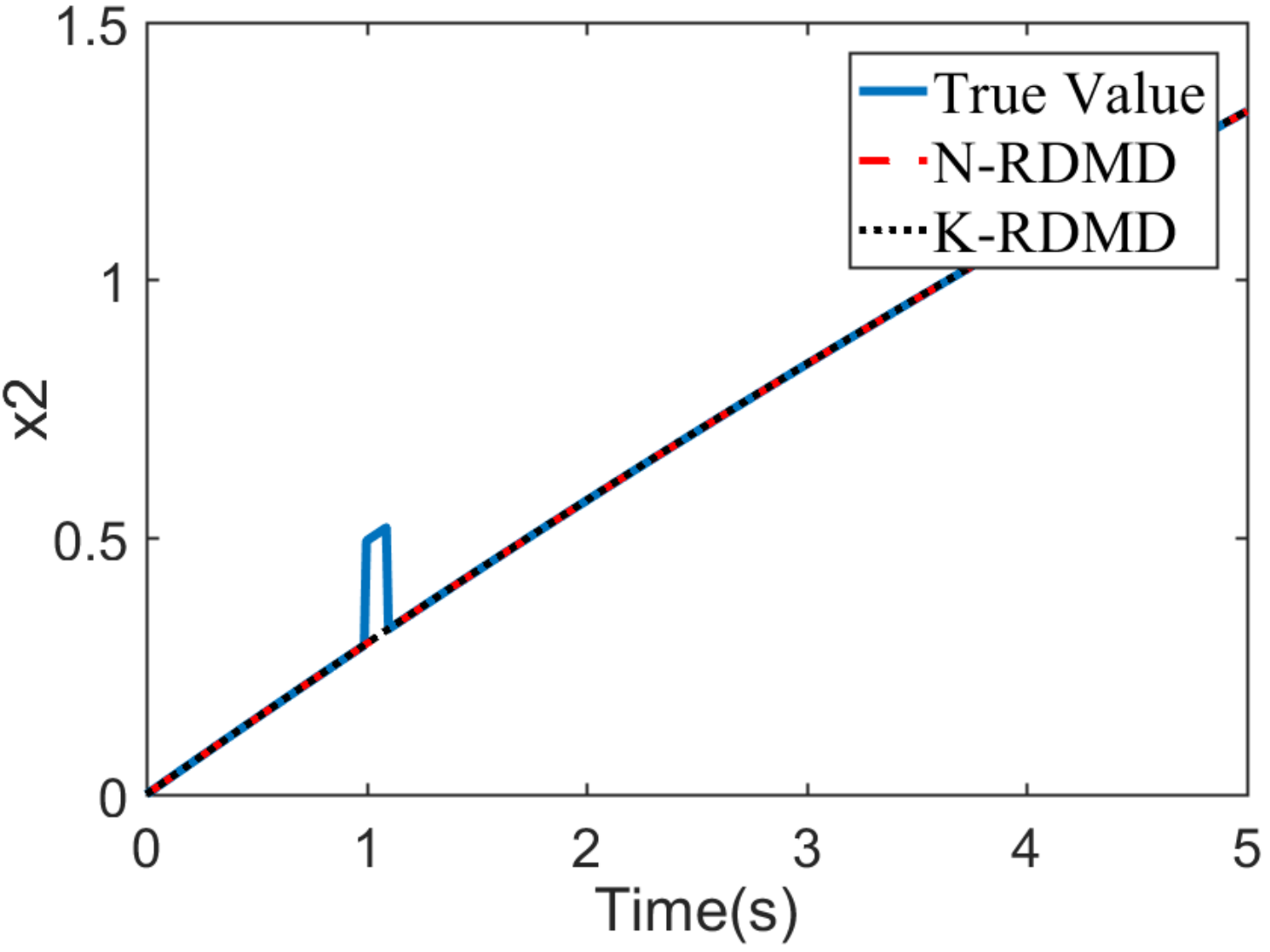}}
\subfloat[]{\includegraphics[width=0.25\textwidth]{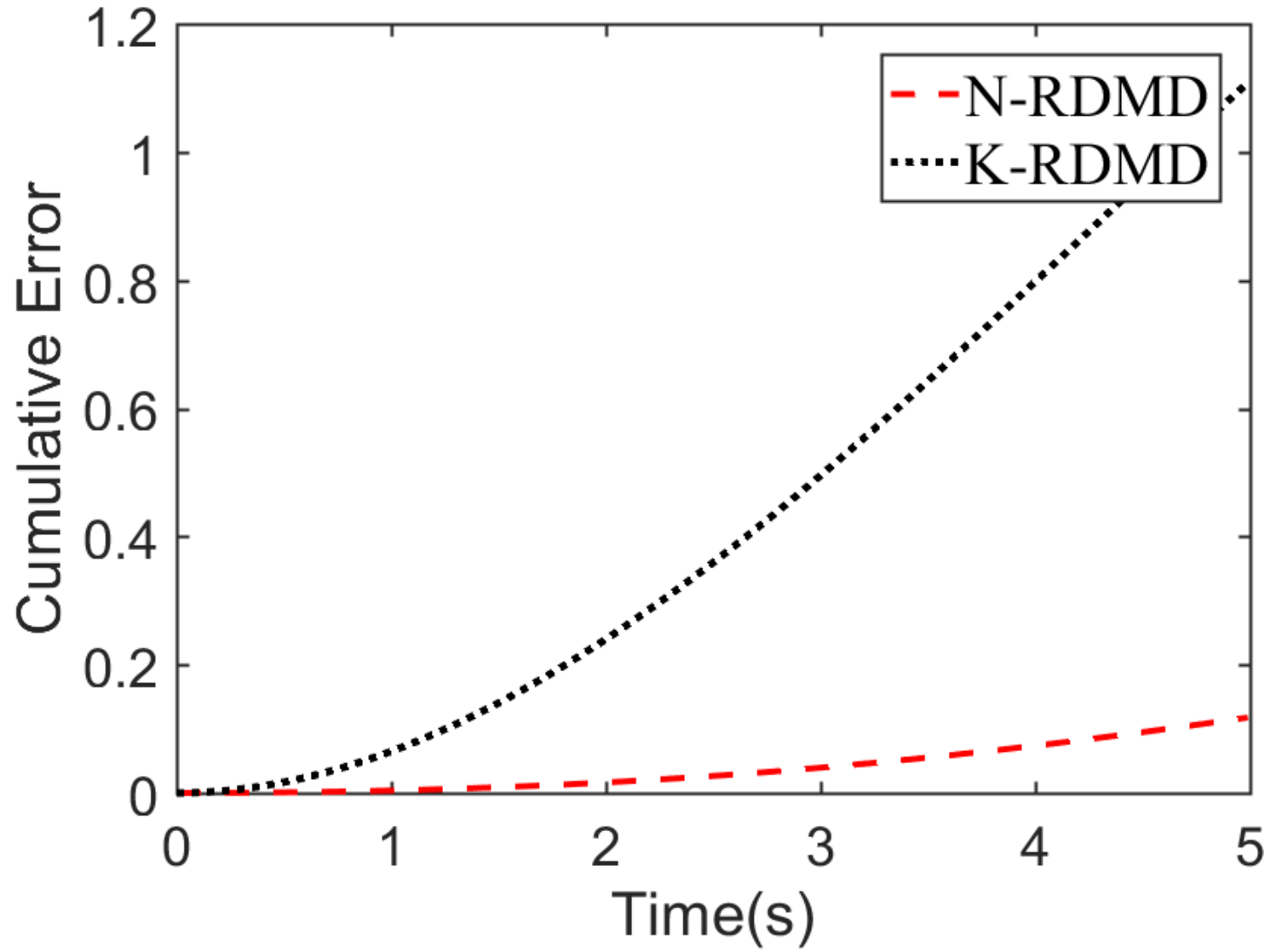}} \\
\caption{Ring network defined in (\ref{eq.ring}) with $s=15$ oscillators and $N=500$ snapshots. (a) Reconstruction of the angular velocity of the first oscillator in the presence of outliers of magnitude $0.2$ from $t=1$ to $t=1.1$ seconds. (b) Cumulative error of the state reconstruction.}
\label{fig.2x}
\end{figure}

\subsection{Linear System}\label{sec.b}
Consider the linear dynamical system governed by

\begin{align}
\dot{x}_{1} &= -x_{1} -3x_{2}, \nonumber \\ 
\dot{x}_{2} &= x_{1} +x_{2}. \label{eq.55x}
\end{align}

\begin{table}[!t]
\centering \scriptsize
\setlength{\tabcolsep}{0.4em}
\caption{Eigenvalues of (\ref{eq.55x}) calculated using various methods under three different scenarios} \vspace{-.5cm}
\begin{tabular}{l r r r} \\ \hline
& Outlier-free case & Case 2 & Case 3 \\ \hline
True value & $0.000\pm{j}1.4142$ & $ 0.000\pm{j}1.4142$ & $ 0.000\pm{j}1.4142$ \\ \hline
DMD \cite{Schmid2010} & $0.0000\pm{j}1.4142$ & $-0.1747\pm{j}1.3787$ & $-0.3574\pm{j}1.3115$ \\
TDMD \cite{Hemati2017} & $0.0000\pm{j}1.4142$ & $-0.0021\pm{j}1.4041$ & $-0.0063\pm{j}1.3643$ \\
ODMD \cite{Sinha2020} & $0.0000\pm{j}1.4142$ & $-0.2626\pm{j}1.3669$ & $-0.4834\pm{j}1.2817$ \\
N-RDMD & $0.0013\pm{j}1.4142$ & $-0.0027\pm{j}1.4131$ & $-0.0043\pm{j}1.4126$ \\ \hline
\end{tabular}
\label{tab.2}
\end{table}

\begin{figure}[!t]
\centering
\subfloat[]{\includegraphics[width=0.25\textwidth]{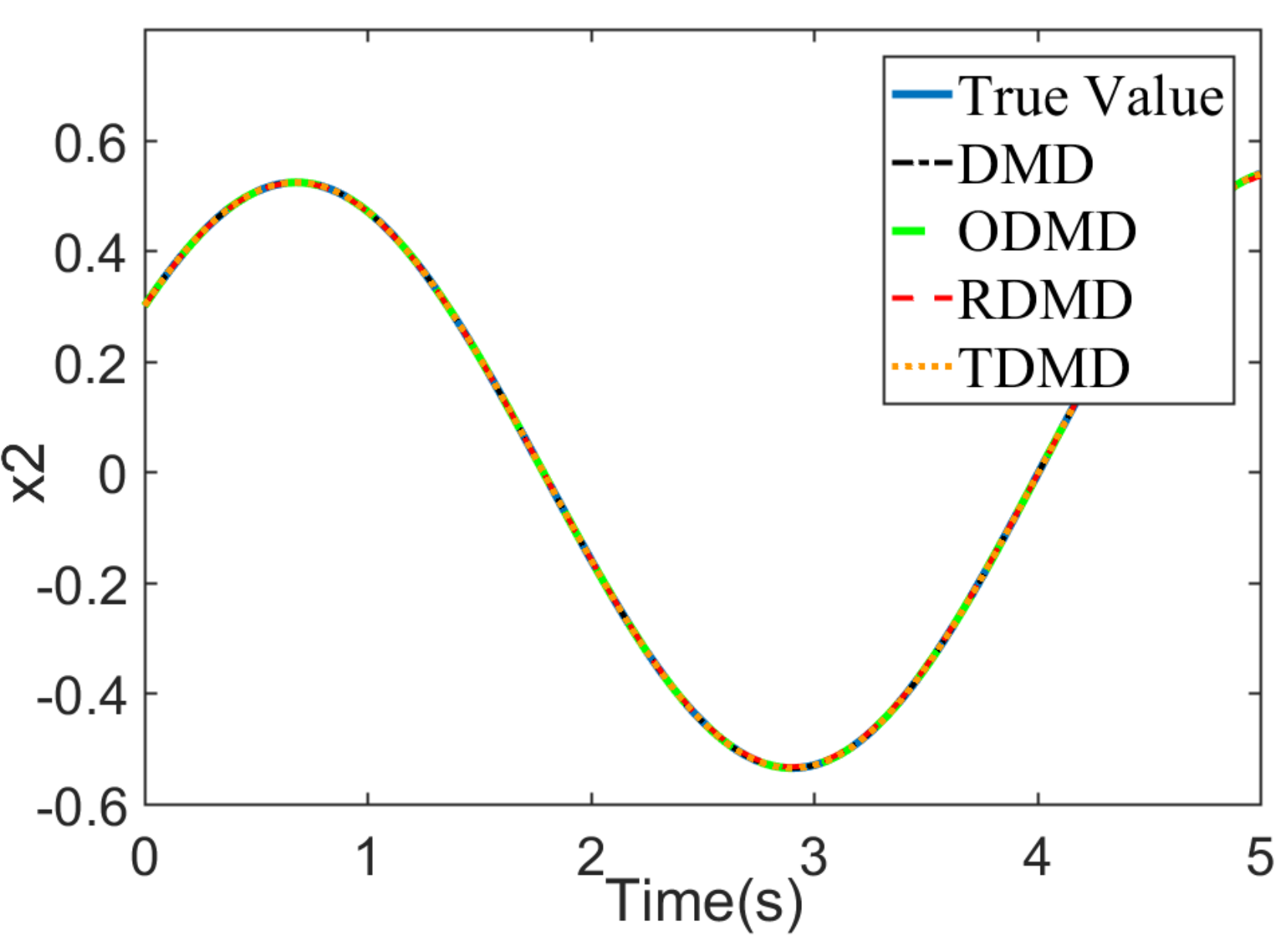}}
\subfloat[]{\includegraphics[width=0.25\textwidth]{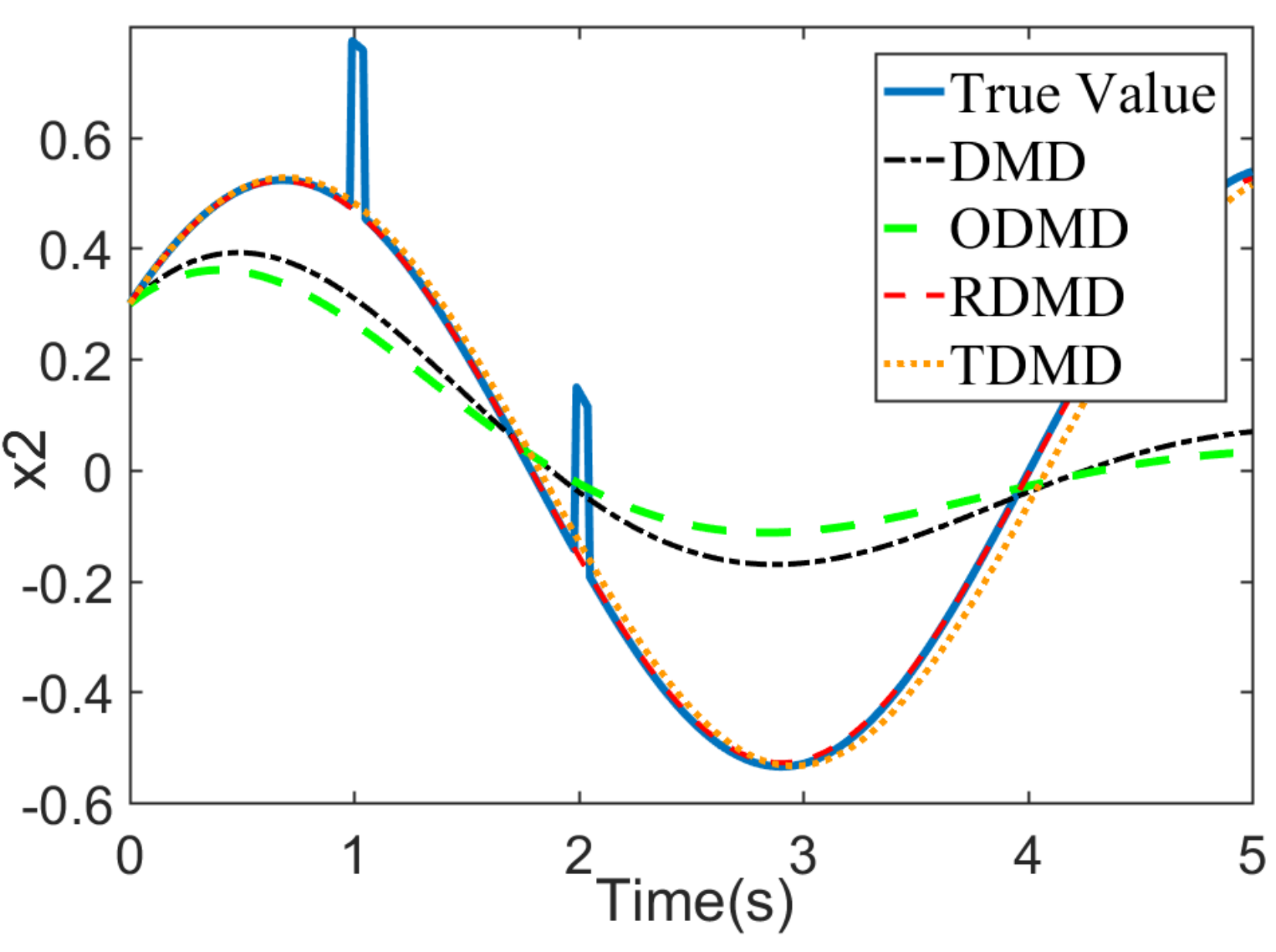}}
\caption{Reconstruction of $x_{2}$ for the linear system in (\ref{eq.55x}). (a) Outlier-free case. (b) Case with outliers of magnitude $0.3$ from $t=1$ to $t=1.05$ seconds and from $t=2$ to $t=2.05$ seconds.}
\label{fig.3x}
\end{figure}

We investigate three cases as follows: 

\begin{enumerate}
\item the sampled data are free of outliers;
\item the sampled data are contaminated with outliers of magnitude $0.3$ from $t=1$ to $t=1.05$ seconds;
\item the sampled data are contaminated with outliers of magnitude $0.3$ from $t=1$ to $t=1.05$ seconds and from $t=2$ to $t=2.05$ seconds.
\end{enumerate}
Table \ref{tab.2} provides the eigenvalues of $\bm{\widehat{A}}$ computed for each case. Furthermore, the reconstruction of $x_{2}$ in Cases 1 and 3 is depicted in Fig. \ref{fig.3x}. Table \ref{tab.2} shows that the accuracy of the RDMD is slightly less than that of other DMD methods in the outlier-free case. As discussed in previous sections, the RDMD presents high statistical efficiency under ideal scenarios while being robust to deviations from assumptions about the data. Indeed, the RDMD performs best for higher percentages of outliers.

\subsection{Network of Coupled Oscillators}
Next, we consider the network of coupled oscillators defined in (\ref{eq.ring}). Figs. \ref{fig.4x} and \ref{fig.5x} show the performance of the considered DMD methods. Fig. \ref{fig.4x} shows that all methods capture the dominant eigenvalue of the system in an ideal case without outliers. Further, all methods yield an accurate reconstruction of state $x_{2}=\omega_{2}$. Conversely, Fig. \ref{fig.5x} shows that the N-RDMD performs best when the sampled data are contaminated with outliers of magnitude $0.1$ between $t=1$ and $t=1.05$ seconds.

\begin{figure}[!t]
\centering
\subfloat[]{\includegraphics[width=0.25\textwidth]{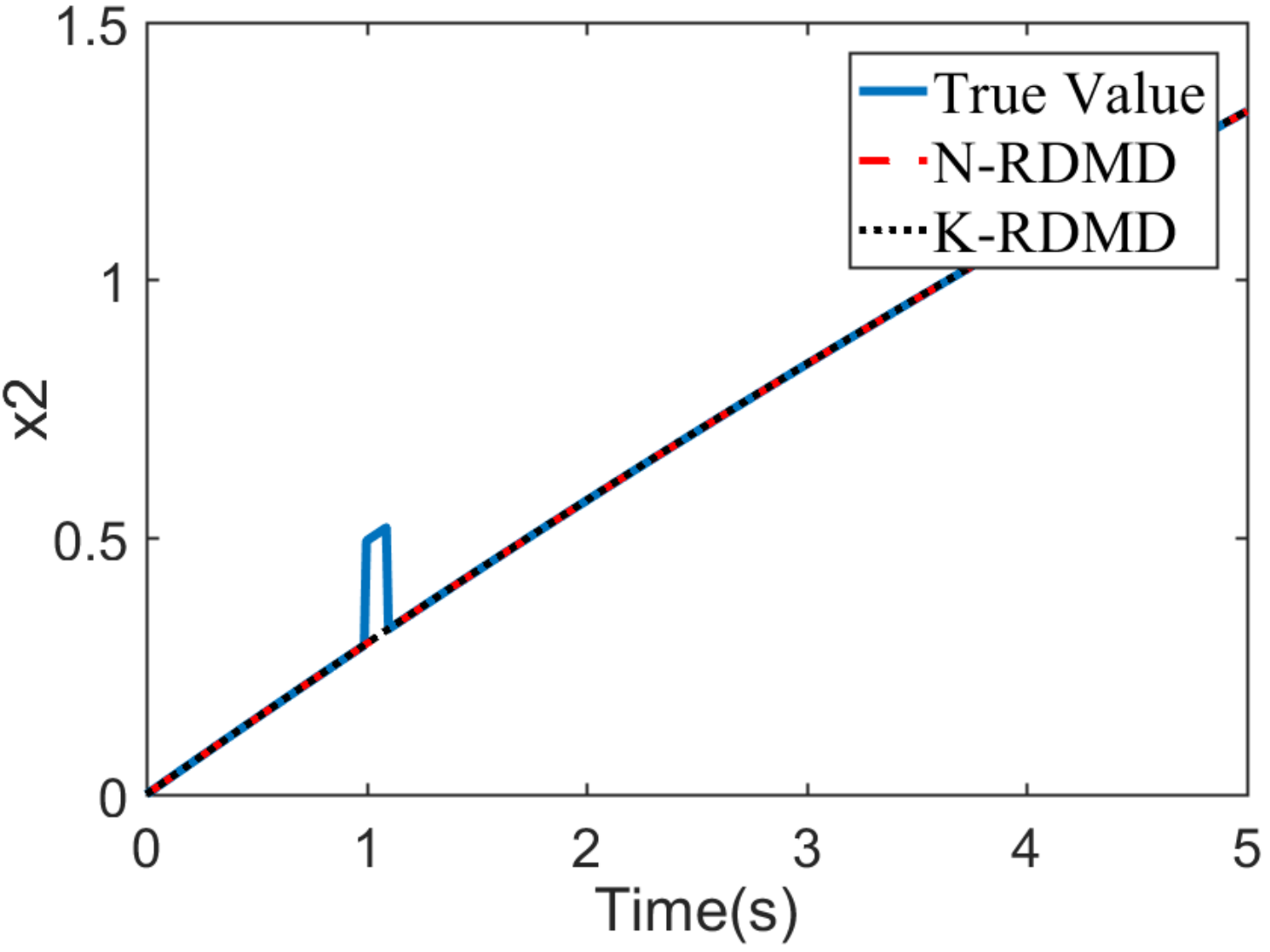}}
\subfloat[]{\includegraphics[width=0.25\textwidth]{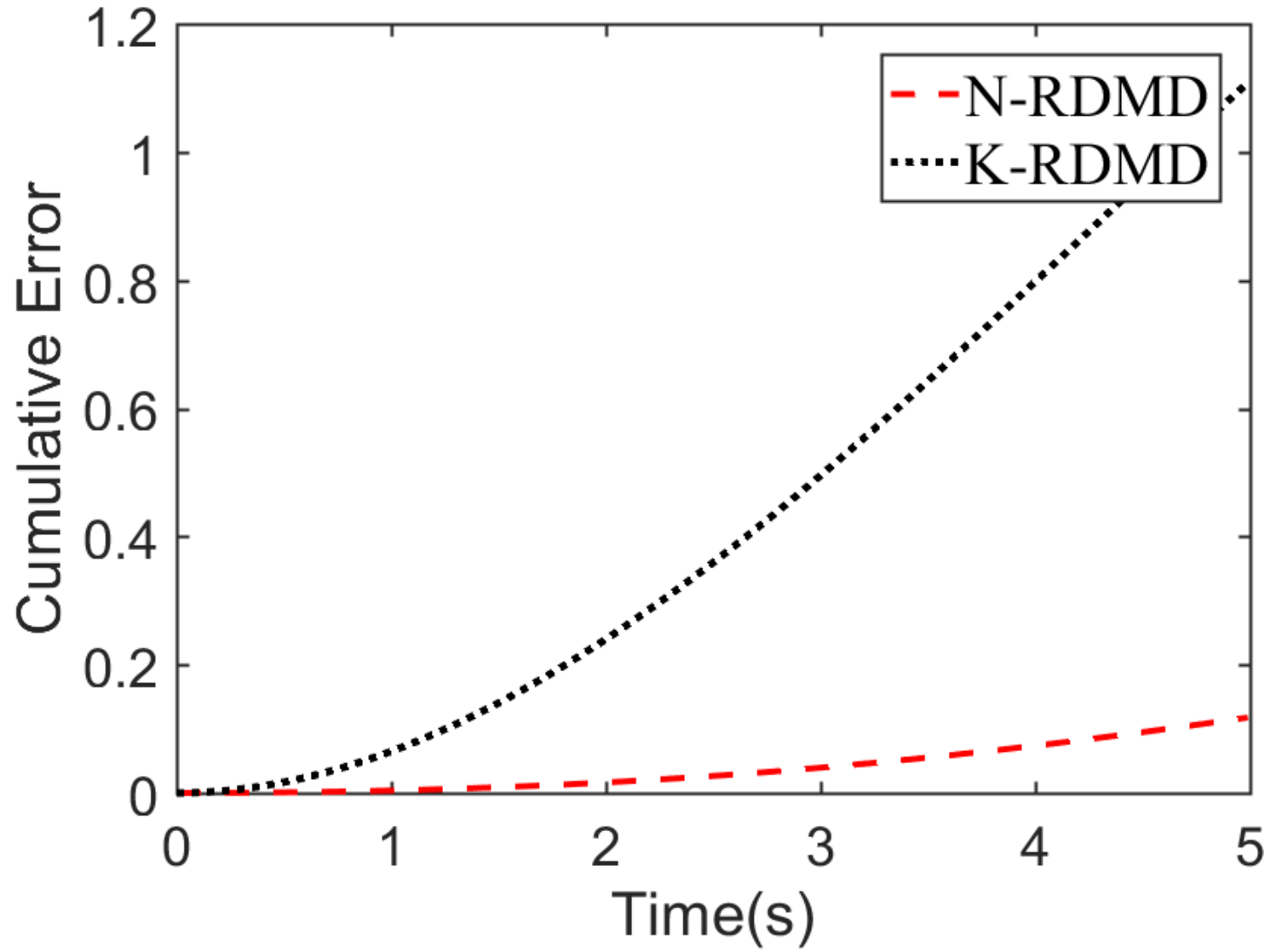}}
\caption{Network of coupled oscillators when the sampled data are outlier free: (a) eigenvalues; (b) reconstruction of the state $x_{2}$.}
\label{fig.4x}
\end{figure}

\begin{figure}[!t]
\centering
\subfloat[]{\includegraphics[width=0.25\textwidth]{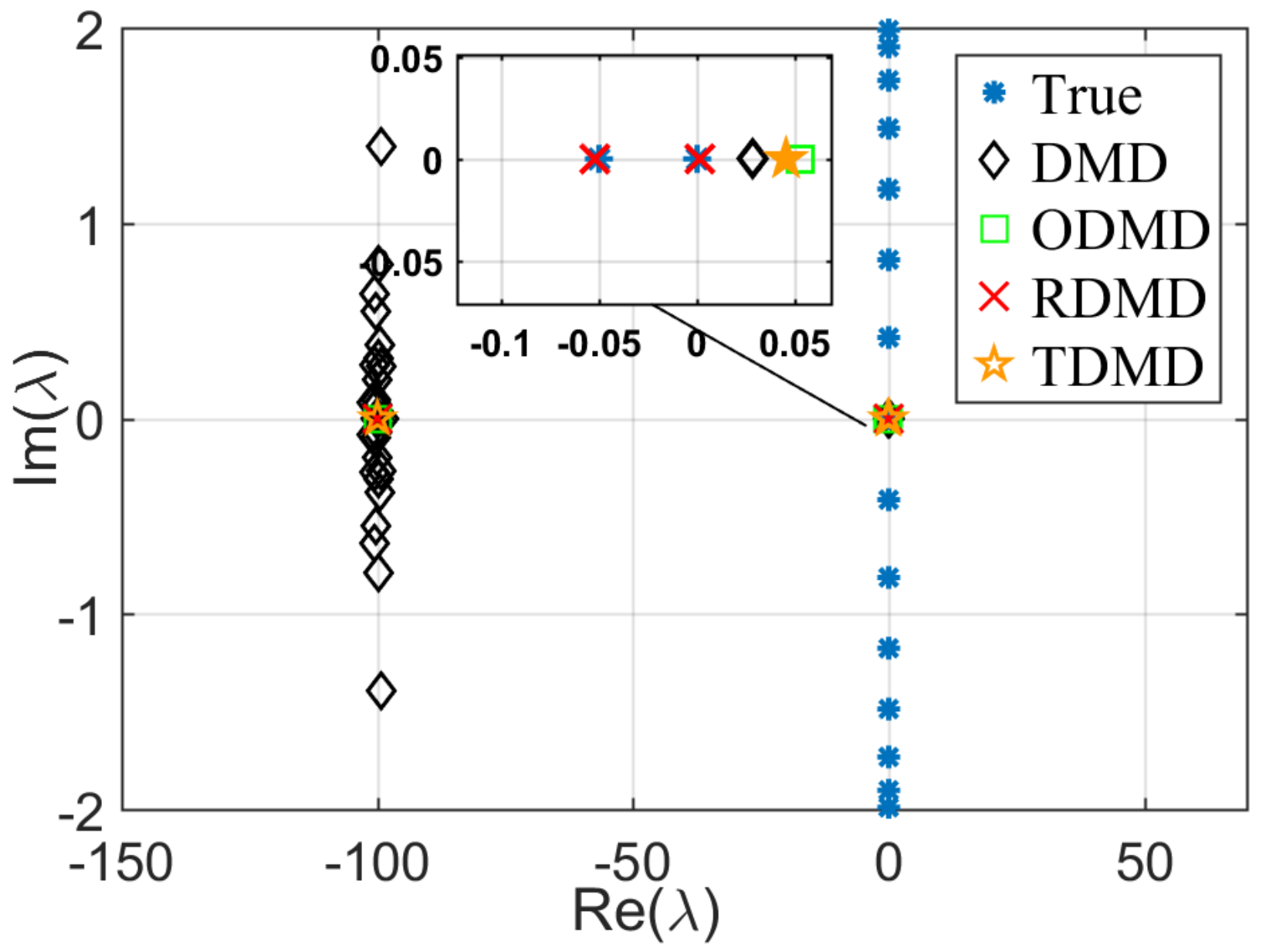}}
\subfloat[]{\includegraphics[width=0.25\textwidth]{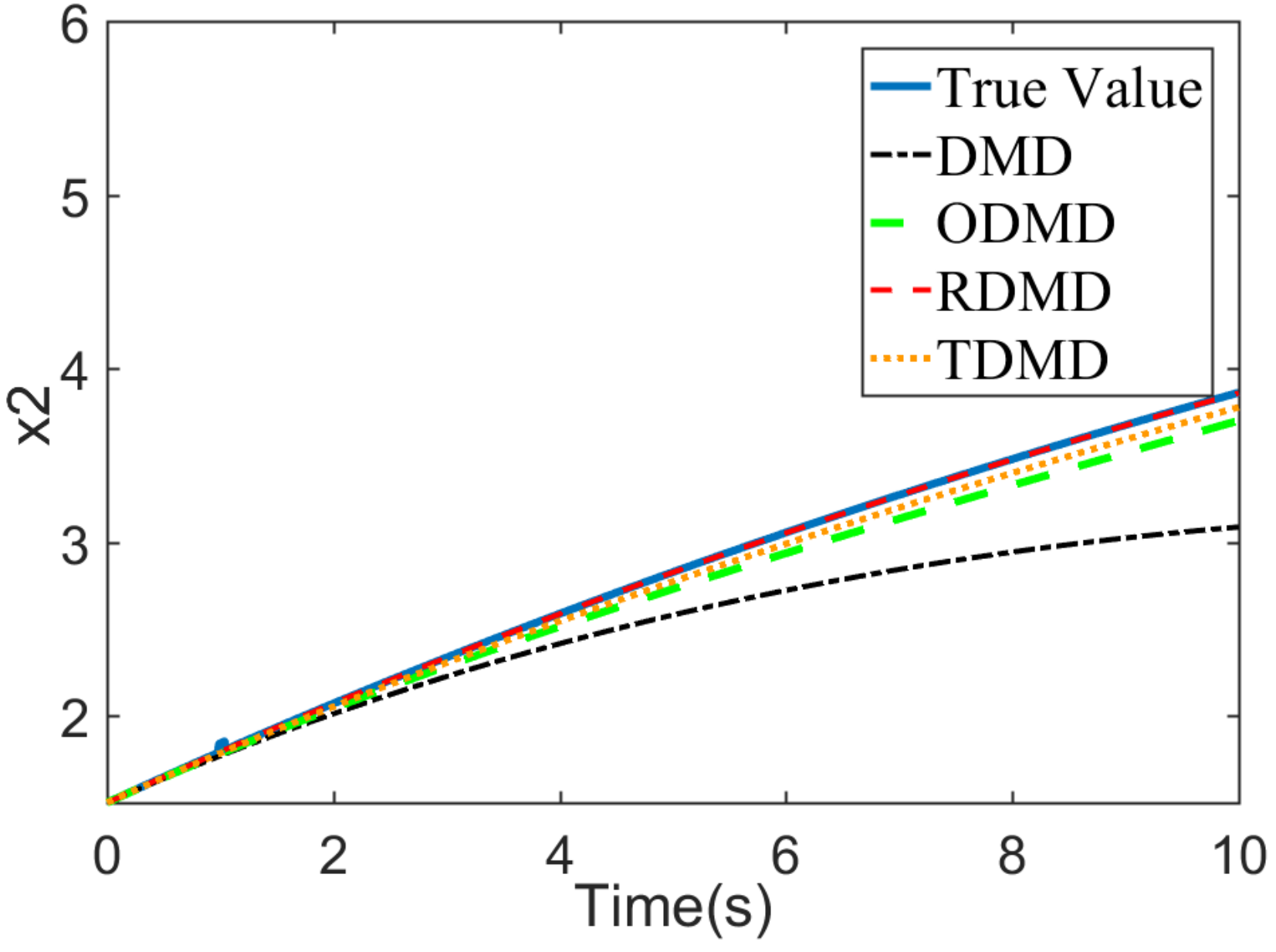}}
\caption{Network of coupled oscillators when the sampled data are contaminated with outliers of magnitude $0.1$ between $t=1$ and $t=1.05$ seconds: (a) eigenvalues; (b) reconstruction of the state $x_{2}$.}
\label{fig.5x}
\end{figure} 

\subsection{Slow-Manifold System}
Now, consider a slow-manifold system given by

\begin{align}
\dot{x}_{1} &= \mu x_{1}, \nonumber \\ 
\dot{x}_{2} &= \lambda \left(x_{2} -\rho(x_{1})\right), \label{eq.59x}
\end{align}

\noindent
with $\mu=-0.05$, $\lambda=-1$, and $\rho(x_{1})=x_{1}^{2}$
{\color{black}(\cite{kaiser2021data}).}
In (\ref{eq.59x}), the slow dynamics are dictated by the eigenvalue equal to $-0.05$, whereas the second state $x_{2}$ quickly approaches the manifold $x_{2}=\rho(x_{1})$. In this case, we simulate a case where the sampled data are contaminated with outliers of magnitude $0.2$ between $t=1$ and $t=1.05$ seconds and between $t=2$ and $t=2.05$ seconds. Fig. \ref{fig.6x} depicts the estimated eigenvalues. We observe that the N-RDMD performs best in capturing the eigenvalue equal to $-0.05$ and that corresponds to the slow dynamics.

\begin{figure}[!t]
\centering
\subfloat[]{\includegraphics[width=0.25\textwidth]{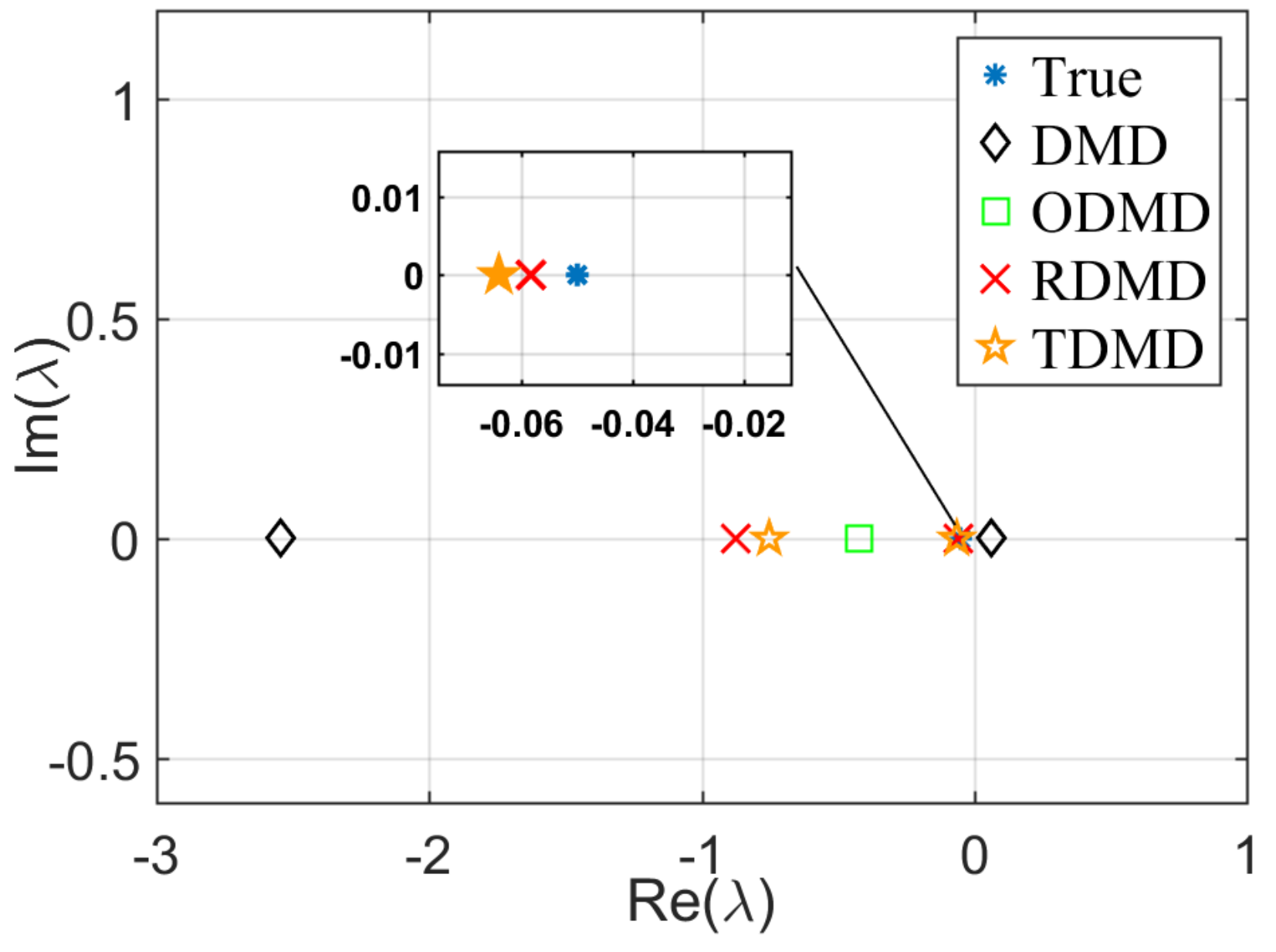}}
\subfloat[]{\includegraphics[width=0.25\textwidth]{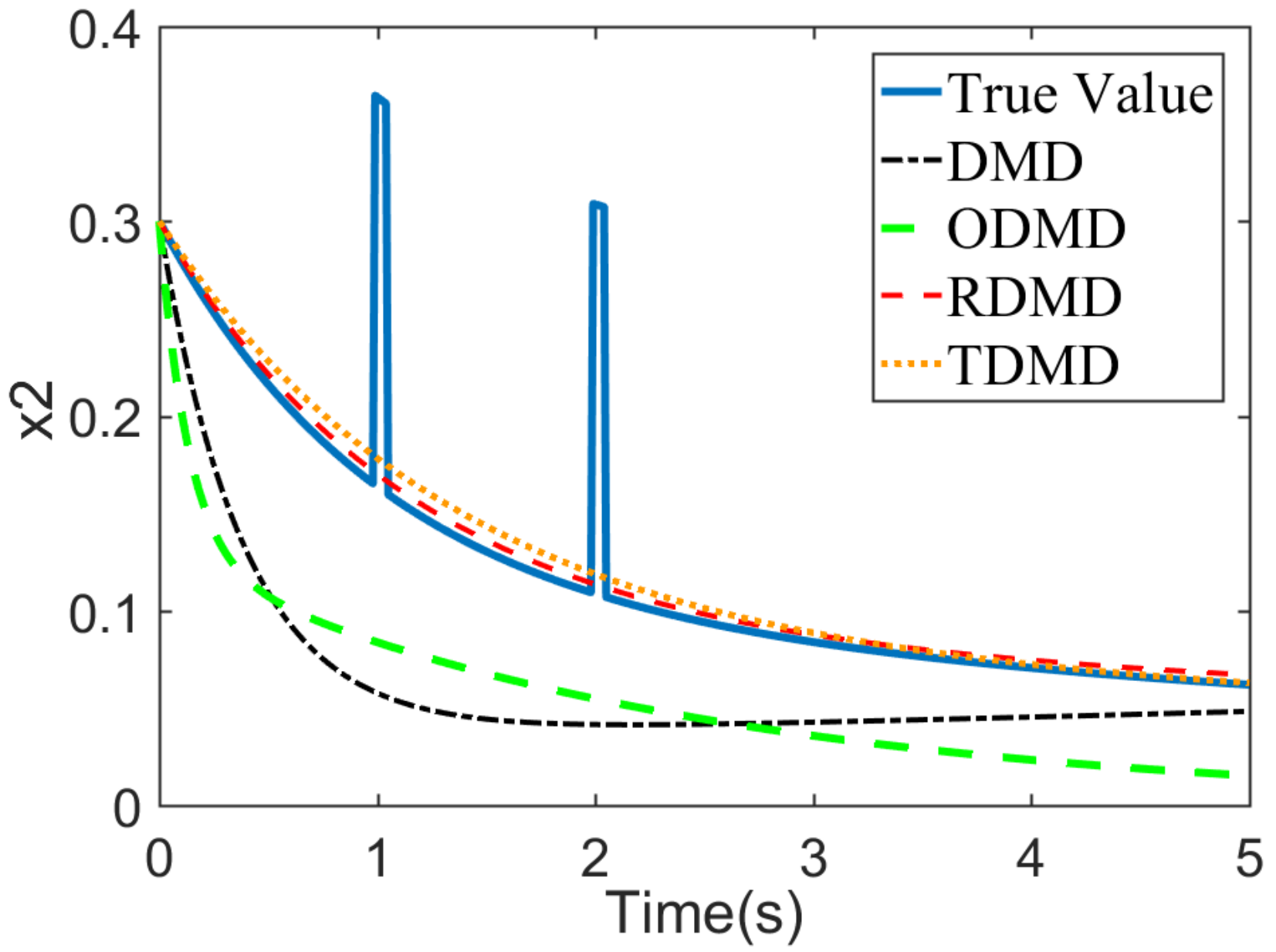}}
\caption{Slow-manifold system when the sampled data are contaminated with outliers of magnitude $0.2$ between $t=1$ and $t=1.05$ seconds and between $t=2$ and $t=2.05$ seconds. (a) Eigenvalues. (b) Reconstruction of state $x_{2}$.}
\label{fig.6x}
\end{figure} 

\begin{figure}[!t]
\centering
\subfloat[]{\includegraphics[width=0.25\textwidth]{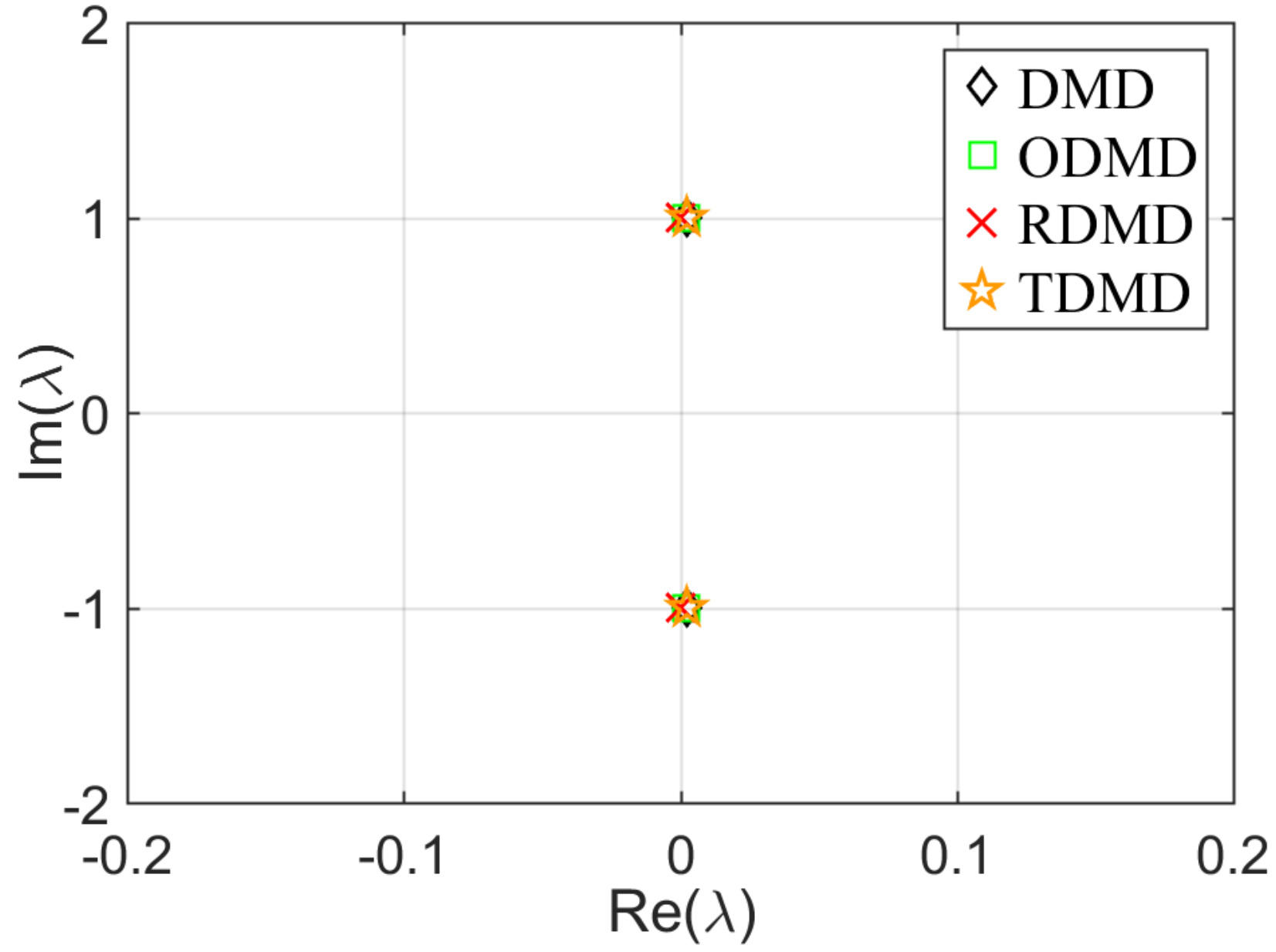}}
\subfloat[]{\includegraphics[width=0.25\textwidth]{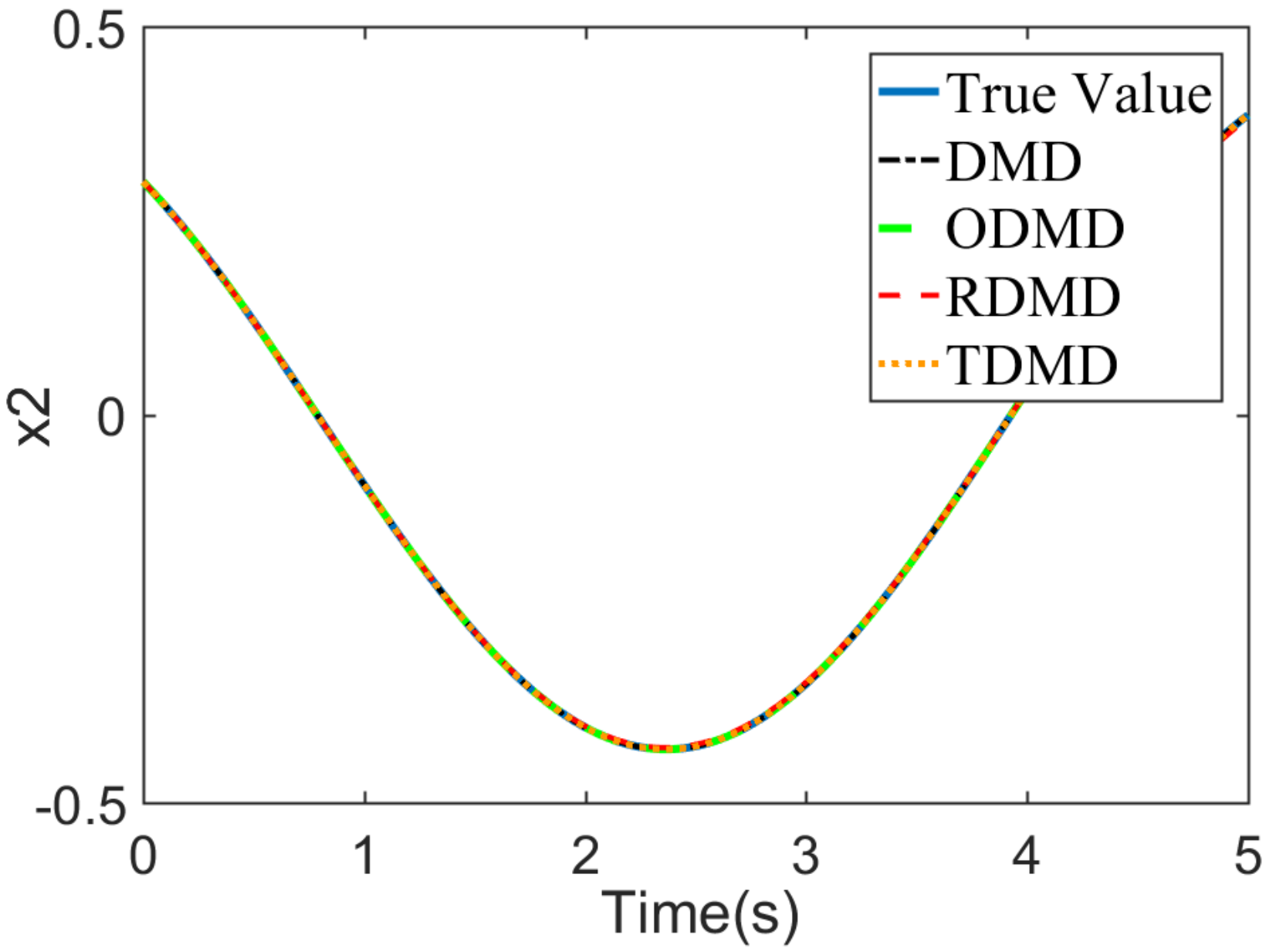}} \\
\subfloat[]{\includegraphics[width=0.25\textwidth]{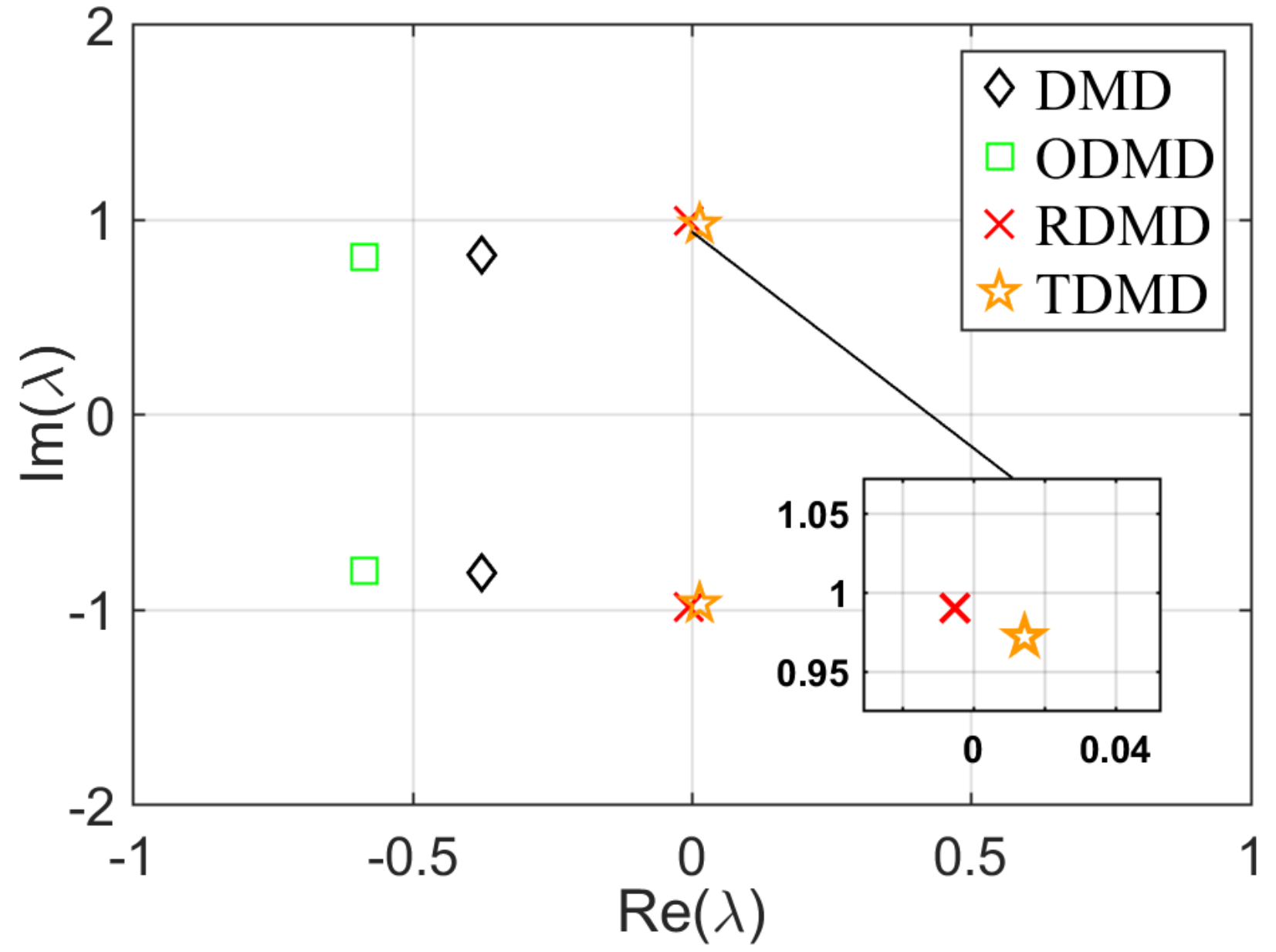}}
\subfloat[]{\includegraphics[width=0.25\textwidth]{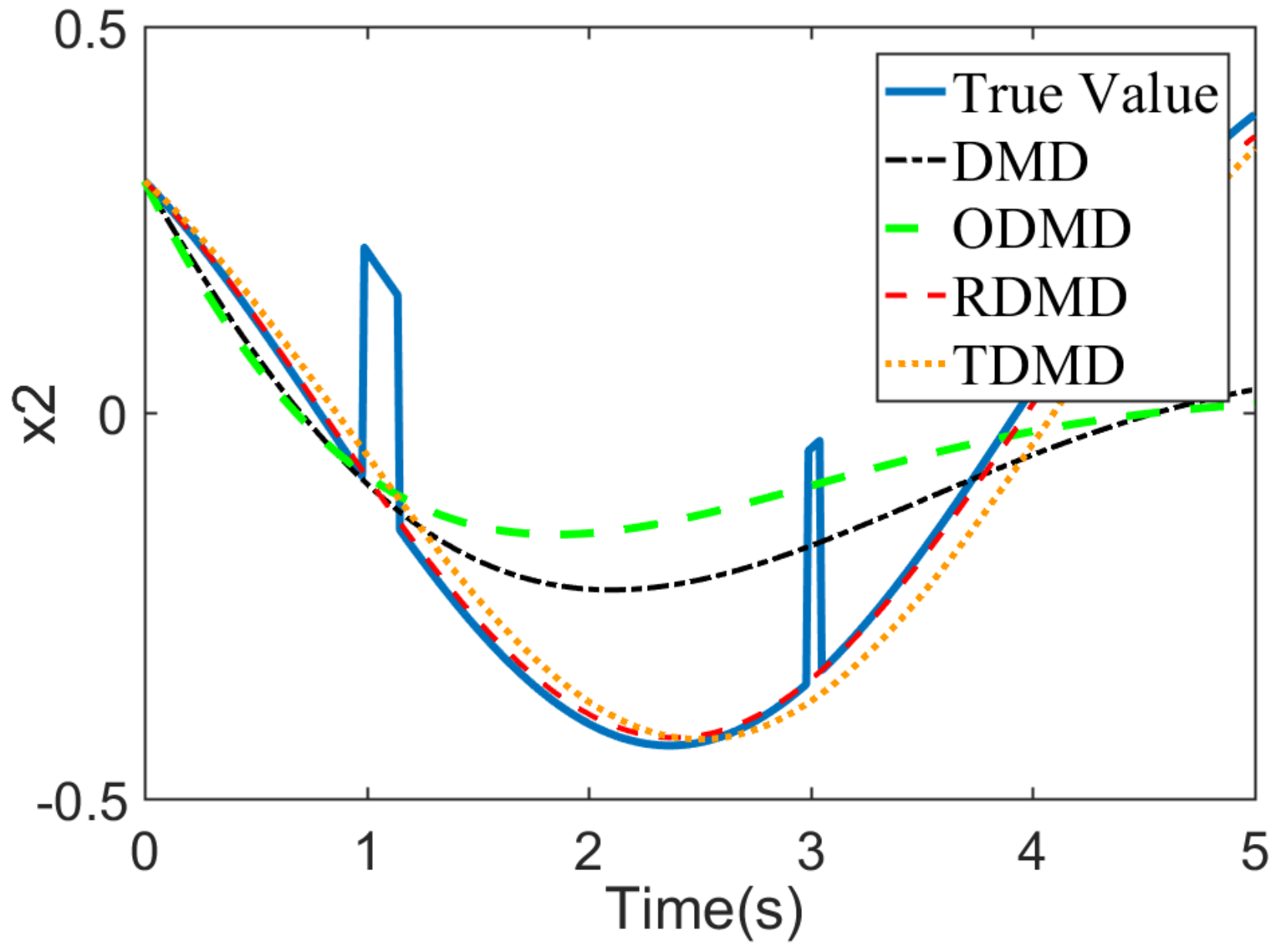}}
\caption{Estimated eigenvalues and reconstruction of state $x_{2}$ for the Van der Pol oscillator in (\ref{eq.VanderPol}):  (a), (b) outlier-free data; (c), (d) sampled data contaminated with outliers of magnitude $0.3$ from $t=1$ to $t=1.15$ seconds and from $t=3$ to $t=3.05$ seconds.}
\label{fig.7x}
\end{figure} 

\subsection{Van der Pol Oscillator}
Finally, we consider the Van der Pol oscillator given by:

\begin{align}
\dot{x}_{1} &= x_{2}, \nonumber \\ 
\dot{x}_{2} &= \mu \left(1-x_{1}^{2}\right)x_{2}, \label{eq.VanderPol}
\end{align}

\noindent
and the results are depicted in Fig. \ref{fig.7x}. Note that for the data collected from nonlinearly evolving signals, the approximated DMD modes are reflecting the behavior of the most dominant Koopman modes. As shown in Fig. \ref{fig.7x}, when there is no outlier contamination within the data set, all considered DMD methods calculate approximately the same dominant eigenvalues, and a good response of $x_{2}$ is reconstructed; however, when there are some outliers among the data set, only N-RDMD can capture the same eigenvalues as found for the outlier-free data. In other words, the process of finding eigenvalues has been made robust against outliers.

\subsection{Dynamic Mode Decomposition of Large Data Sets}\label{sec.f}
Next, we assess the performance of the N-RDMD for larger data sets. First, we consider a random linear system of the form given by

\begin{equation}\label{eq.RandomLinear}
\bm{\dot{x}}(t)=\bm{\Theta}\bm{x}(t),
\end{equation}

\noindent
where $\bm{\Theta}$ is a random matrix of the form $\text{randn}(m,m)-h\bm{I}_m$, and $h$ is chosen such that all eigenvalues reside in the left half of the complex Cartesian plane. The data are collected for $N=200$ time samples. The order of the truncated dynamics is considered to be $c^{\prime}=25$ for all the simulated DMD methods. Note that ODMD \cite{Sinha2020} is not included in this case because no solution can be attained. Fig. \ref{fig.8x} shows that the exact DMD and TDMD methods perform best when the data are outlier free; however, the N-RDMD outperforms other methods when the data are contaminated with outliers. Note that the TDMD estimates the eigenvalues with positive real parts, indicating  unstable dynamics; therefore, we do not plot the TDMD response in Fig. \ref{fig.8x}(d). 

As a second example, consider a generalized representation of a slow-manifold nonlinear system given by

\begin{align}
\dot{\bm{x}}_{1} &= \bm{W}\bm{x}_{1}, \nonumber \\ 
\dot{\bm{x}}_{2} &= \bm{\Lambda} \left(\bm{x}_{2} -\bm{P}\left(\bm{x}_{1}\right)\right), \label{eq.GeneralizedSlowManifold}
\end{align}

\noindent
where $\bm{x}_{1}$, $\bm{x}_{2}\in\mathbb{R}^{m/2}$, $\bm{W}=\text{diag}\left(\mu_{1},\,...,\,\mu_{m/2}\right)$, $\bm{\Lambda}=\text{diag}\left(\lambda_{1},\,...,\,\lambda_{m/2}\right)$, $\mu_{i}<0$, $\lambda_{i}>0$, $i=\{1,...,m/2\}$, and $\bm{P}(\bm{x}_{1})= [\bm{P}_{1}(\bm{x}_{1})\; ...\; \bm{P}_{m/2}(\bm{x}_{1})]\tran$. Here, we choose $\mu_{i}$ as negative random numbers, $\bm{\Lambda}=\bm{I}_{m/2}$, $\bm{P}_{i}(\bm{x}_{1})=\left(\sum_{j=1}^{m/2}\bm{x}_{1j}\right)^{2}$, and $m=500$.

The results depicted in Fig. \ref{fig.9x} demonstrate the ability of all DMD methods to capture the dominant eigenvalues, located at $\mu_{i}$, and to reconstruct the state $x_{2}$ when the data are outlier free. Further, as shown in Fig. \ref{fig.9x}, the N-RDMD outperforms the DMD and the TDMD when the sampled data are contaminated with outliers of magnitude $0.1$ from $t=1$ to $t=1.05$ seconds.

\begin{figure}[!t]
\centering
\subfloat[]{\includegraphics[width=0.25\textwidth]{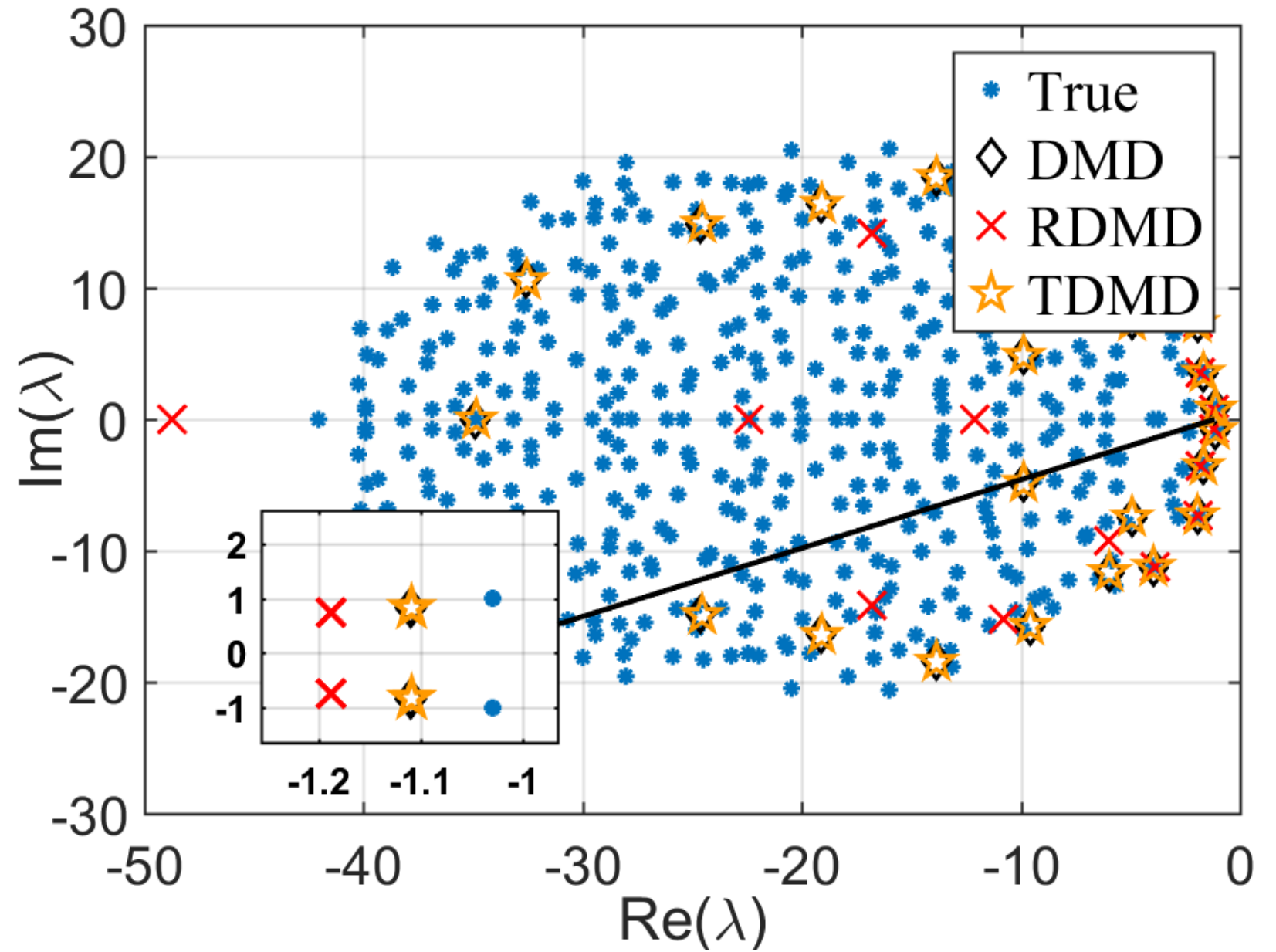}}
\subfloat[]{\includegraphics[width=0.25\textwidth]{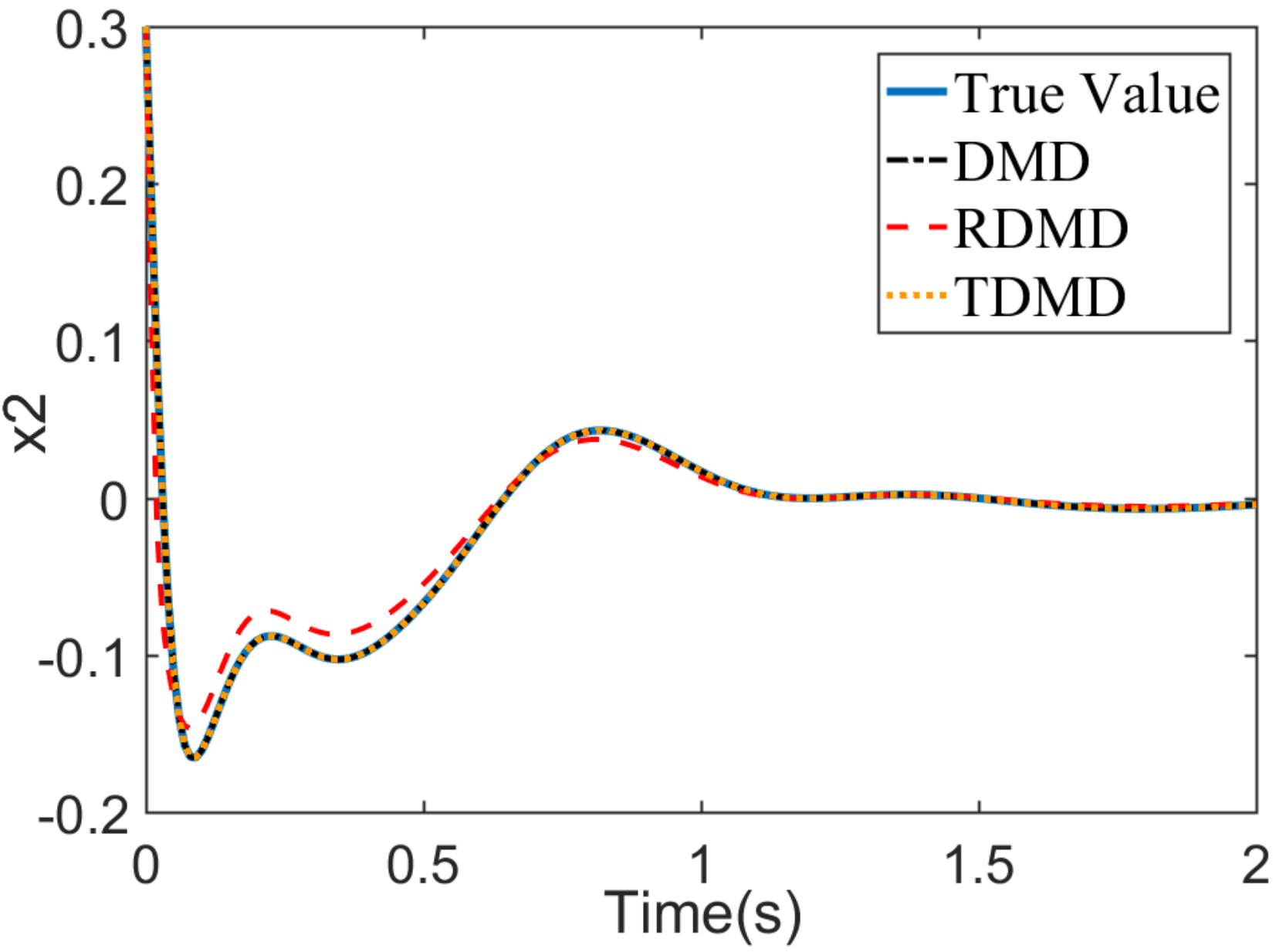}} \\
\subfloat[]{\includegraphics[width=0.25\textwidth]{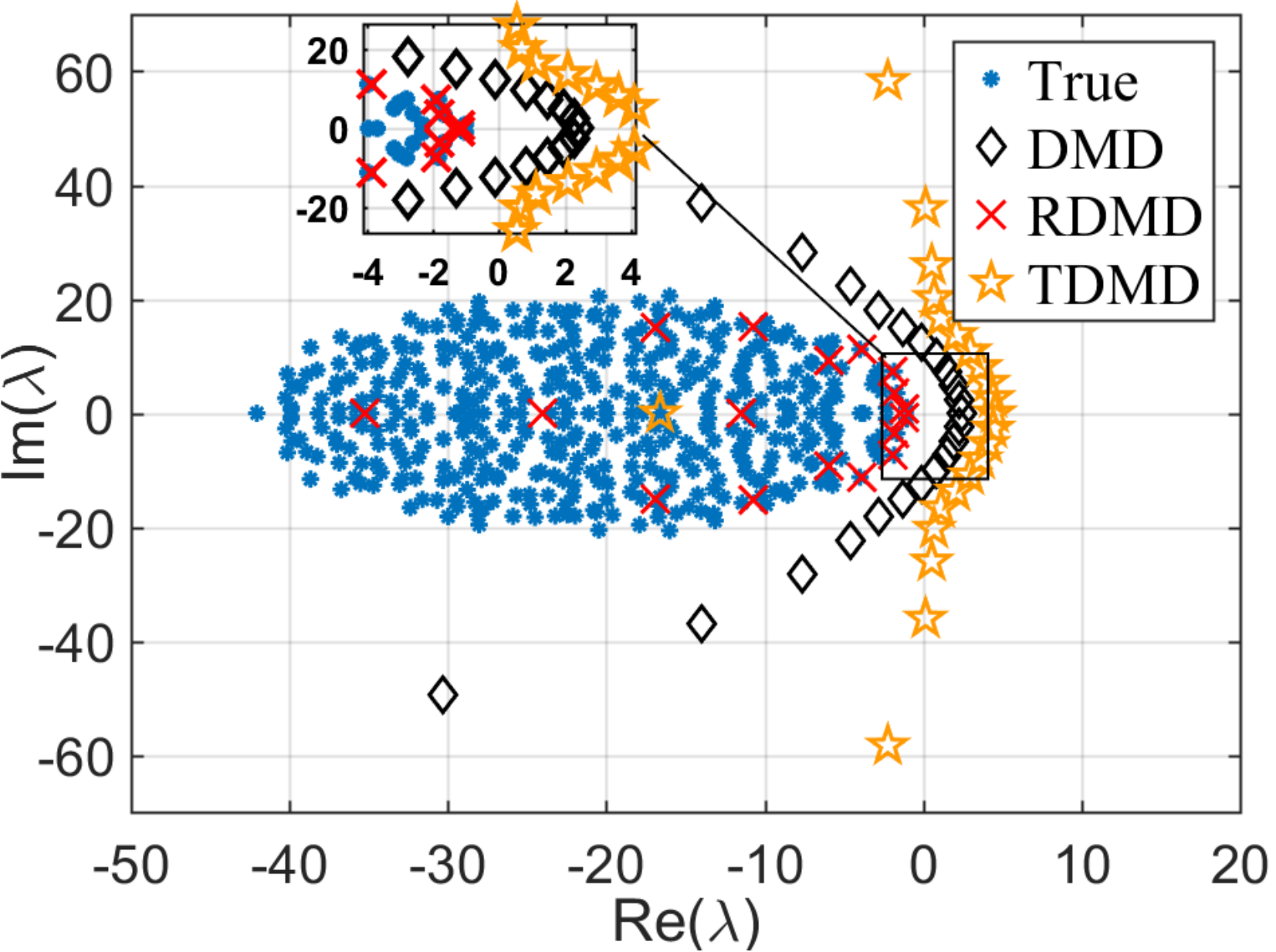}}
\subfloat[]{\includegraphics[width=0.25\textwidth]{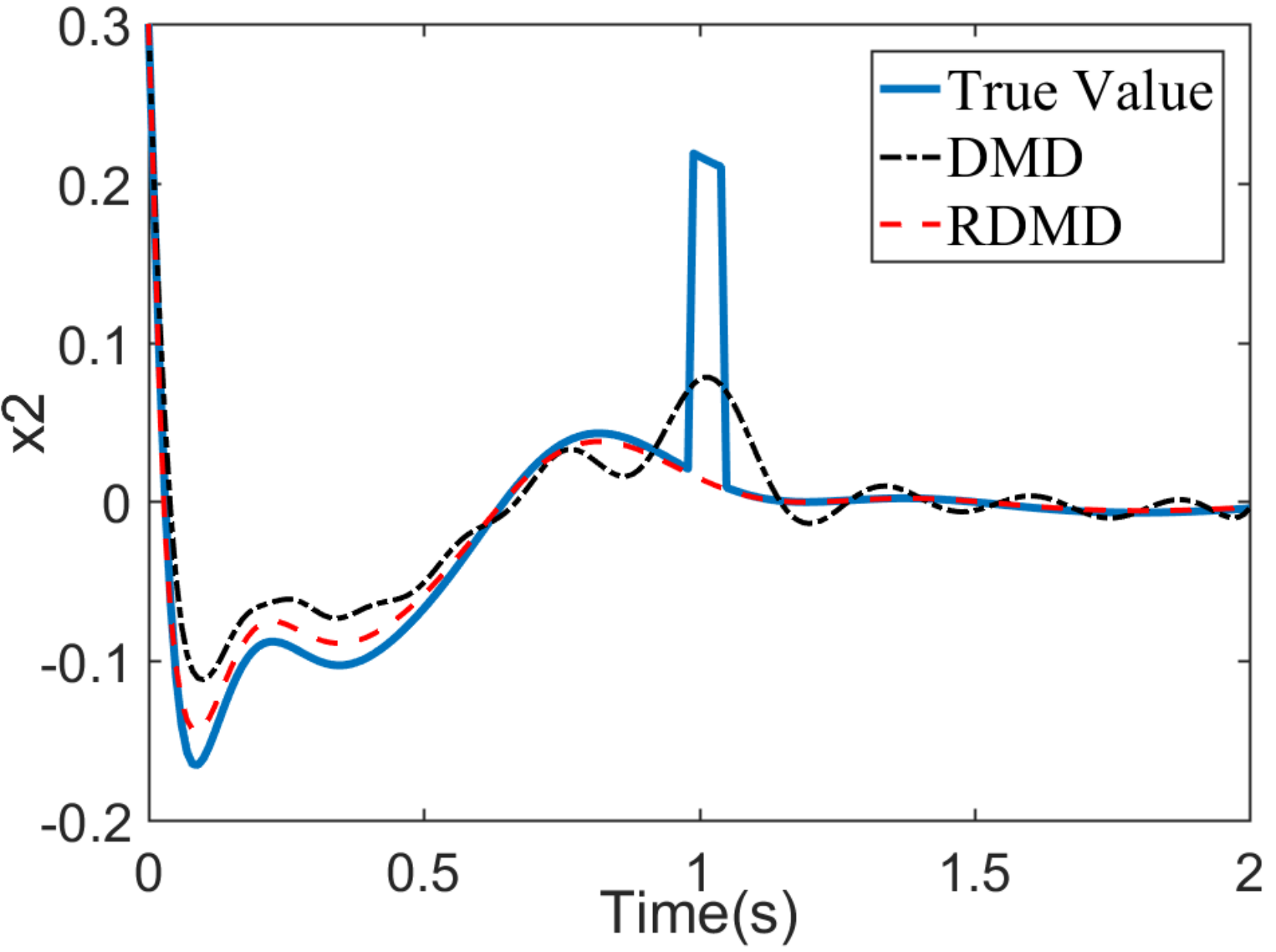}}
\caption{Estimated eigenvalues and reconstruction of state $x_{2}$ for the random linear dynamical system in (\ref{eq.RandomLinear}): (a), (b) outlier-free data; (c), (d) sampled data contaminated with outliers of magnitude $0.2$ from $t=1$ to $t=1.05$ seconds.}
\label{fig.8x}
\end{figure}

\begin{figure}[!t]
\centering
\subfloat[]{\includegraphics[width=0.25\textwidth]{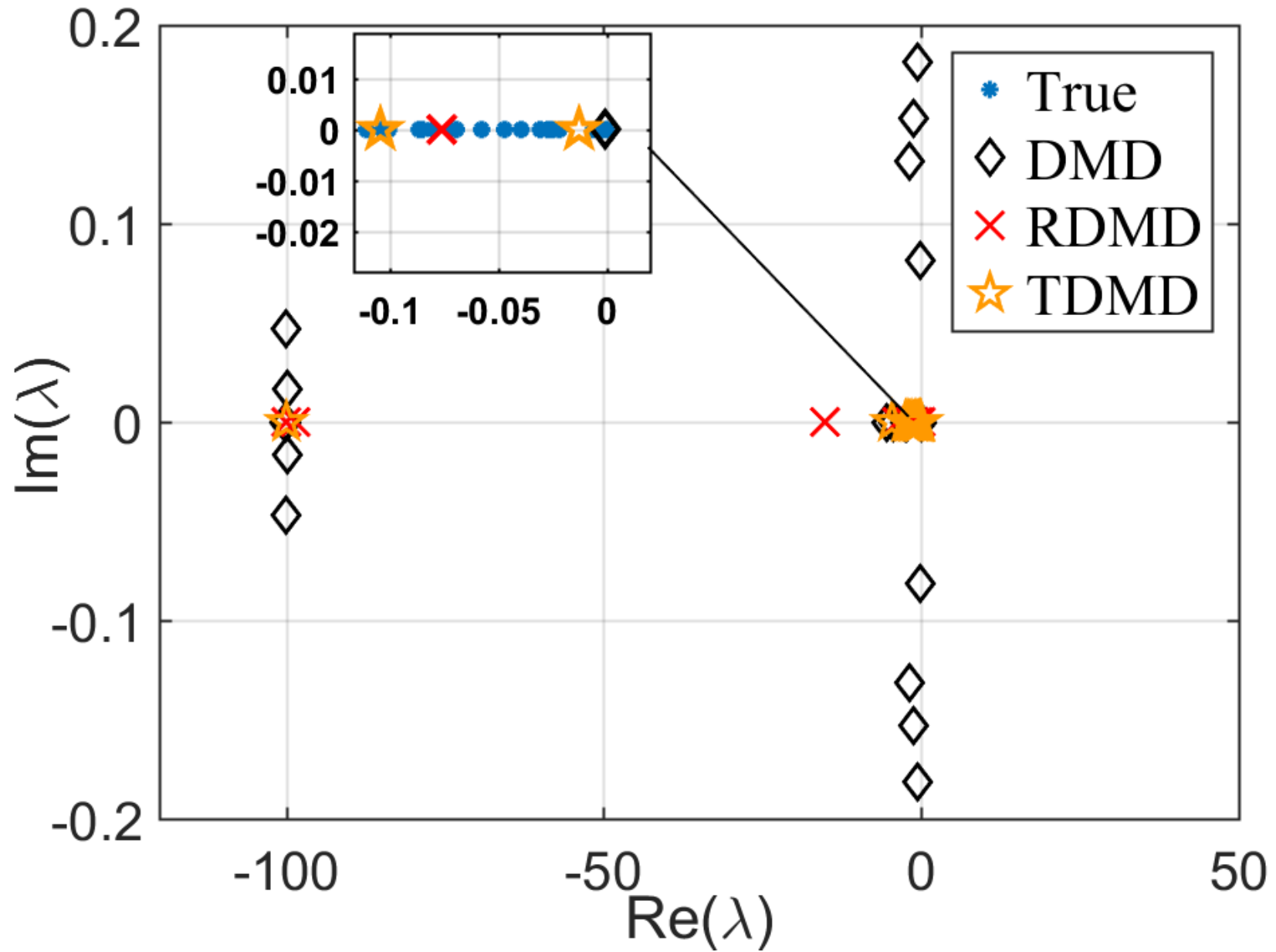}}
\subfloat[]{\includegraphics[width=0.25\textwidth]{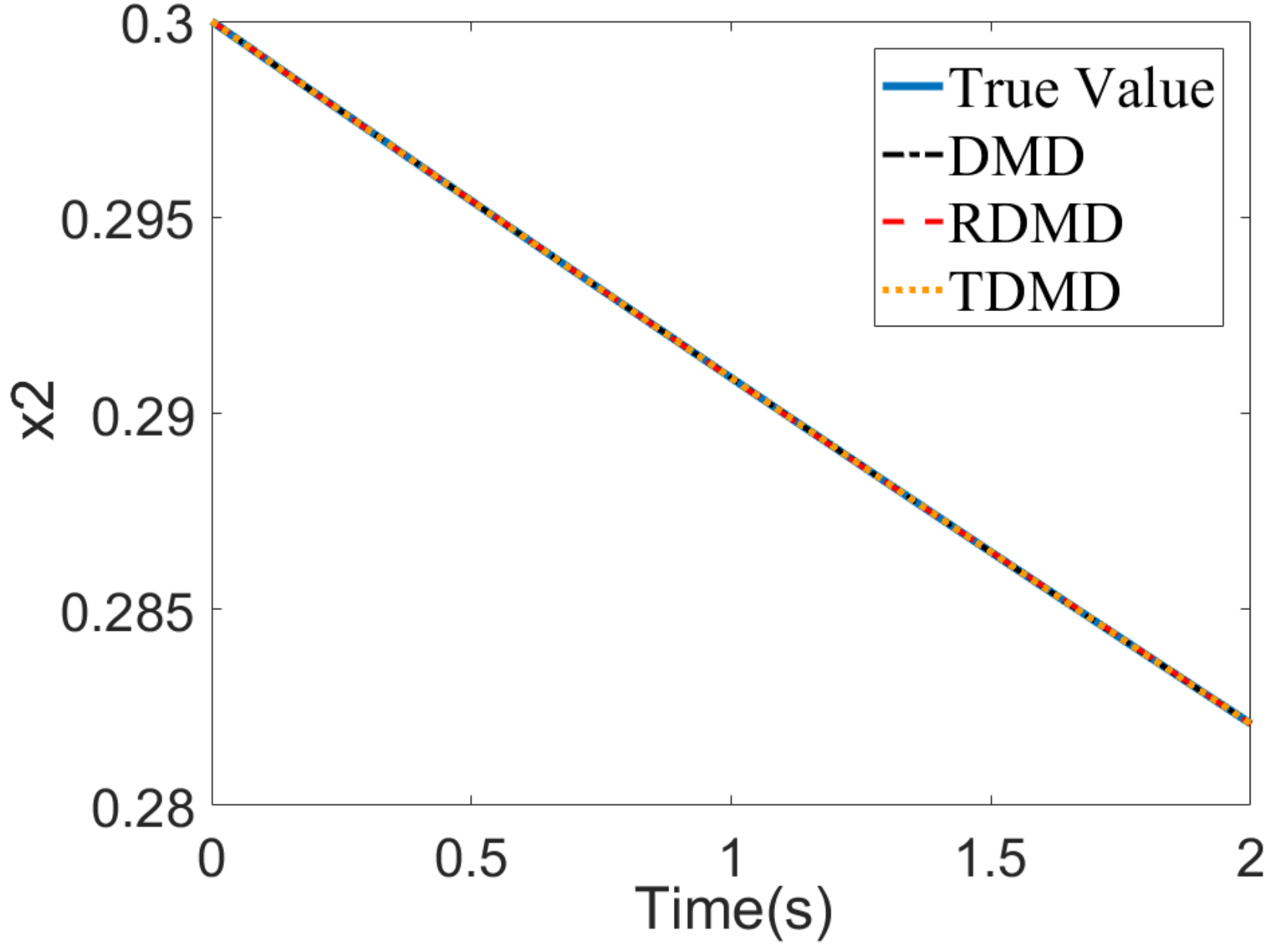}} \\
\subfloat[]{\includegraphics[width=0.25\textwidth]{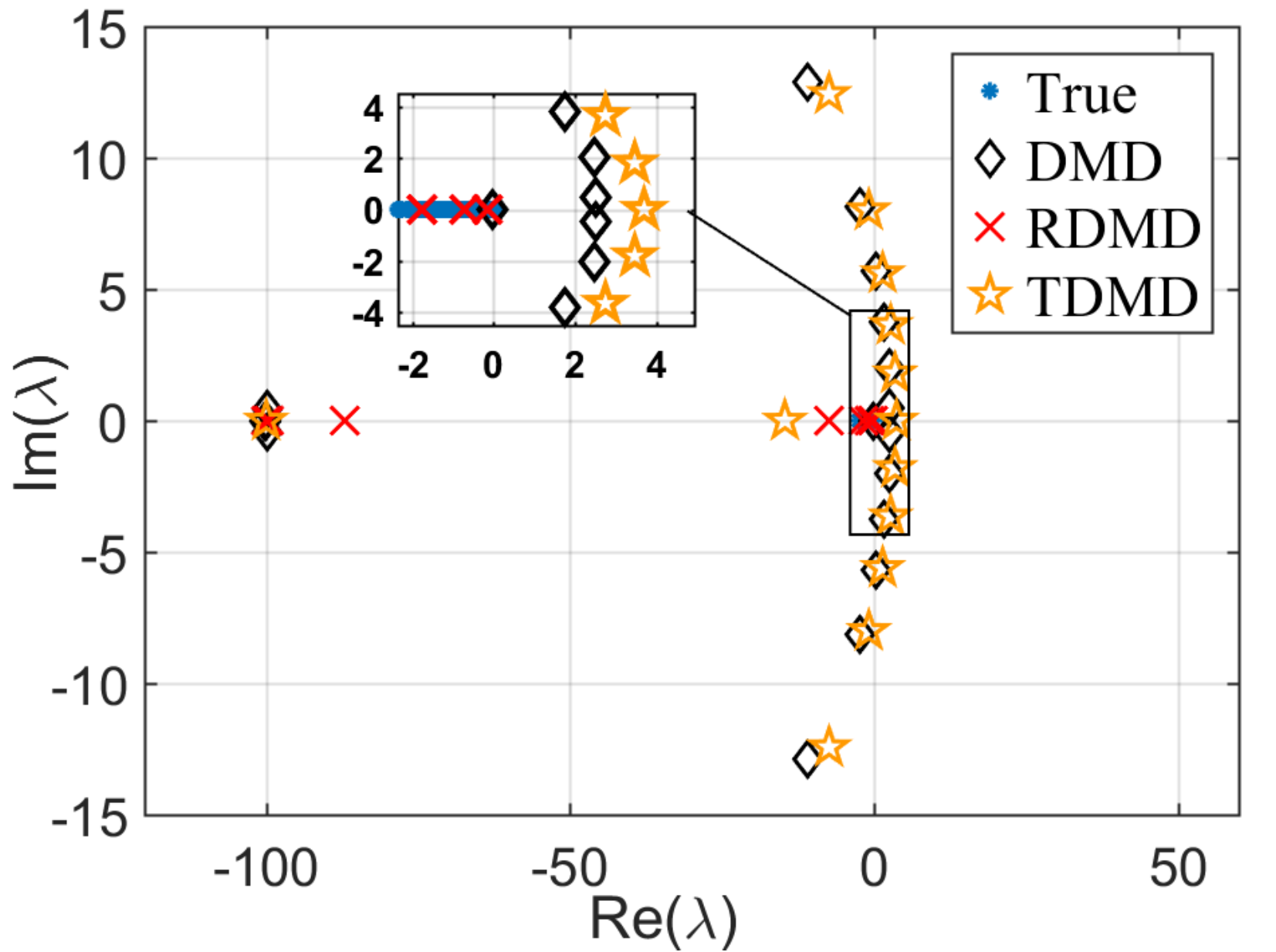}}
\subfloat[]{\includegraphics[width=0.25\textwidth]{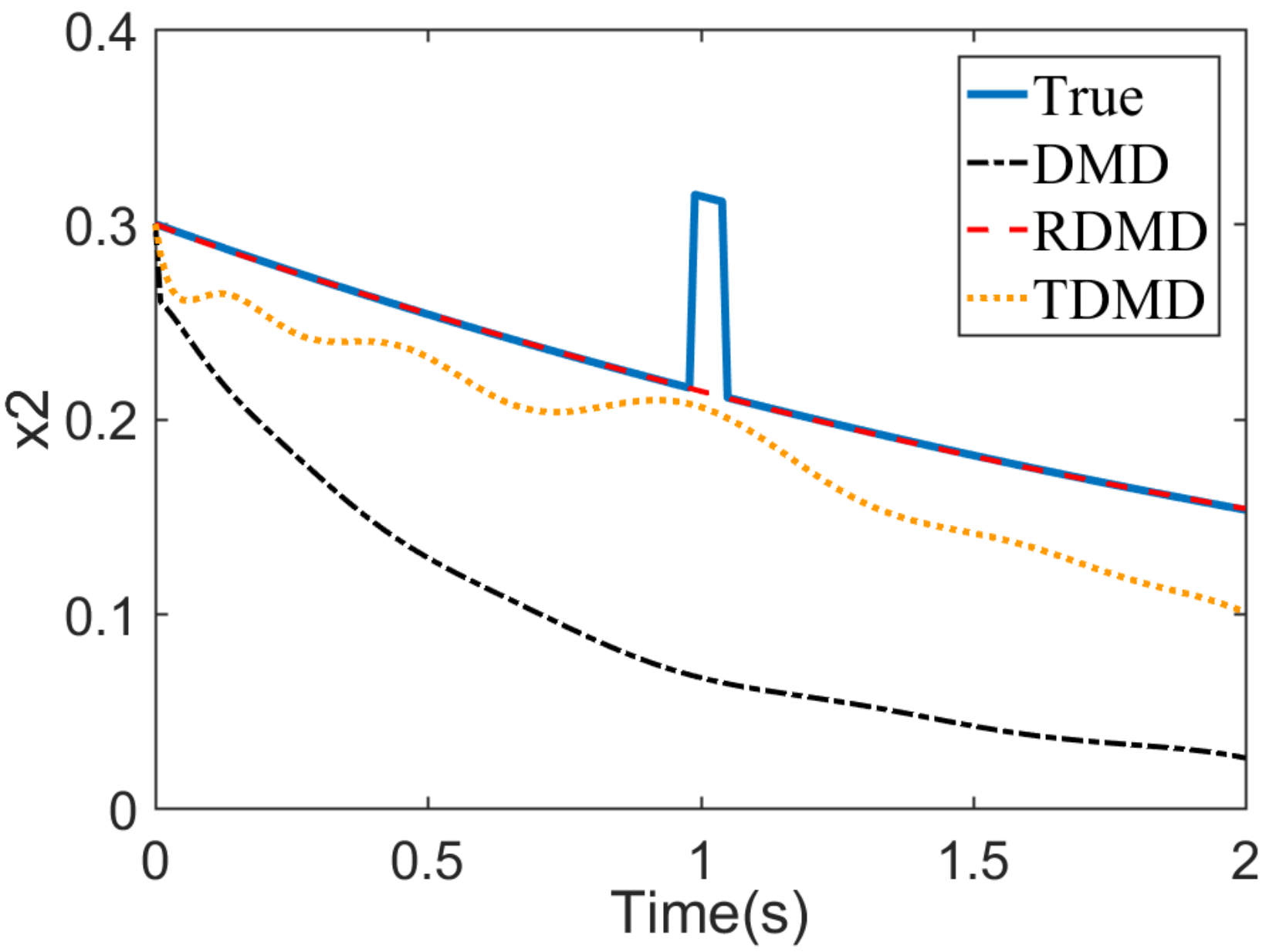}}
\caption{Estimated eigenvalues and reconstruction of state $x_{2}$ for the generalized slow-manifold system in (\ref{eq.GeneralizedSlowManifold}): (a), (b) outlier-free data;  (c), (d) Sampled data contaminated with outliers of magnitude $0.2$ from $t=1$ to $t=1.05$ seconds.}
\label{fig.9x}
\end{figure}

\subsection{Computation Time}
A comparison of the computation time for all DMD methods applied on the examples in sections \ref{sec.b} to \ref{sec.f} is provided in Table \ref{tab.3}. Compared to other methods, the N-RDMD has a higher computation time, the vast majority of which is spent computing projection statistics. This is essentially the price to pay for having statistical robustness. 

\begin{table}[!t]
\centering \scriptsize
\setlength{\tabcolsep}{0.7em}
\caption{Computation time (s) of the tested DMD methods} \vspace{-.5cm}
\begin{tabular}{l c c c c c c} \\ \hline
Method & \textit{B} & \textit{C} & \textit{D} & \textit{E} & \textit{F1} & \textit{F2} \\ \hline
DMD \cite{Schmid2010} & $0.008$ & $0.022$ & $0.007$ & $0.008$ & $0.024$ & $0.027$ \\
TDMD \cite{Hemati2017} & $0.010$ & $0.018$ & $0.009$ & $0.010$ & $0.041$ & $0.040$ \\
ODMD \cite{Sinha2020} & $1.092$ & $2.130$ & $1.108$ & $1.117$ & $-$ & $-$ \\
N-RDMD & $4.173$ & $2.926$ & $3.198$ & $4.012$ & $1.058$ & $1.072$ \\ \hline
\end{tabular}
\label{tab.3}
\end{table}

\subsection{Gaussian and Non-Gaussian Noises}\label{sec.non-Gaussian}
The proposed RDMD has two important practical advantages: i) it can suppress the adverse effect of outliers, and ii) it is statistically efficient even if the residues do not follow a Gaussian distribution; see, e.g., \cite{Mili1996r, Netto2018, Netto2018b, Zhao2019}. In this section, we assess the performance of the proposed N-RDMD method in the presence of (a) Gaussian, (b) Laplace, (c) Student-$t$, and (d) Cauchy noise. The results are depicted in Fig. \ref{fig.10x}. The TDMD method performs best for all noise types, followed by the proposed N-RDMD method. 

The superior performance of the TDMD is attributed to the symmetric debiasing used by such a method, which makes it the best choice for noisy data. TDMD considers the uncertainties in both $\bm{Y}$, $\bm{Y}^{\prime}$, so the problem becomes $\bm{Y}^{\prime}+\Delta\bm{Y}^{\prime}=\bm{A} (\bm{Y}+\Delta\bm{Y})$, whereas other methods consider only the uncertainties on the right-hand side of the problem. This explains why the TDMD method performs best for noisy data and, to a certain extent, suppresses the effect of outliers in the estimation process. The combination of the TDMD and the proposed N-RDMD methods is promising and will be addressed in future research.

\subsection{Comparison with Robust Least-Trimmed Square Dynamic Mode Decomposition}
As discussed, a numerical method that is statistically robust was proposed in \cite{Askham2017}. This method works based on an LTS estimator; thus, we refer to it as LTS-RDMD. The discussion on the efficiency and applicability of the LTS estimators is given in the introduction. As stated, generalized maximum-likelihood estimators are easier to implement and faster to calculate as compared to LTS estimators. Consider a two-dimensional dynamical system governed by 
\begin{equation}
\dot{\bm{x}}\left( t \right)=\left[ \begin{matrix}
1 & -2 \\
1 & -1 \\
\end{matrix} \right]\bm{x}\left( t \right)
\label{eq_simple_oscillator}
\end{equation}
The measurement snapshots are contaminated with the additive deviation $\eta \bm{w}\left( t \right)+\mu \bm{s}\left( t \right)$, where $\bm{w}(t)$ is the Gaussian noise. Also, the elements of the vector $\bm{s}(t)$ are obtained by multiplying a Bernoulli trial with small expectation $p$ by a standard normal, which leads to a sparse noise; therefore, the snapshots are contaminated with a base Gaussian noise and some spikes of size $\mu$ with firing rate $p$ (\cite{Askham2017}). The comparison between our RDMD and the LTS-RDMD to discover the eigenvalues of the system is given in Fig. \ref{fig.11x}. Note that we repeated the simulations 200 times for each method. Simulations are performed for fixed ${\mu}=1$ and $p=0.05$ and for two different Gaussian noises levels: $\sigma=10^{-4}$ and $\sigma=10^{-3}$. It is observed that both methods can efficiently discover the true eigenvalues when the data are contaminated by the spikes (outliers). Also, in a number of simulations, the LTS-RDMD finds some erroneous eigenvalues, whereas the RDMD always finds them near the true value. Moreover, Fig. \ref{fig.11x} shows that the LTS-RDMD method acts more precisely for higher levels of the Gaussian noise. The price for such additive robustness is the higher implementational and computational complexity of the LTS-RDMD method. As stated in the previous subsection, the combination of the TDMD and the proposed RDMD methods could further improve the robustness in cases of more powerful Gaussian noise, and it will be addressed in our future research.

\begin{figure}[!t]
\centering
\subfloat[]{\includegraphics[width=0.25\textwidth]{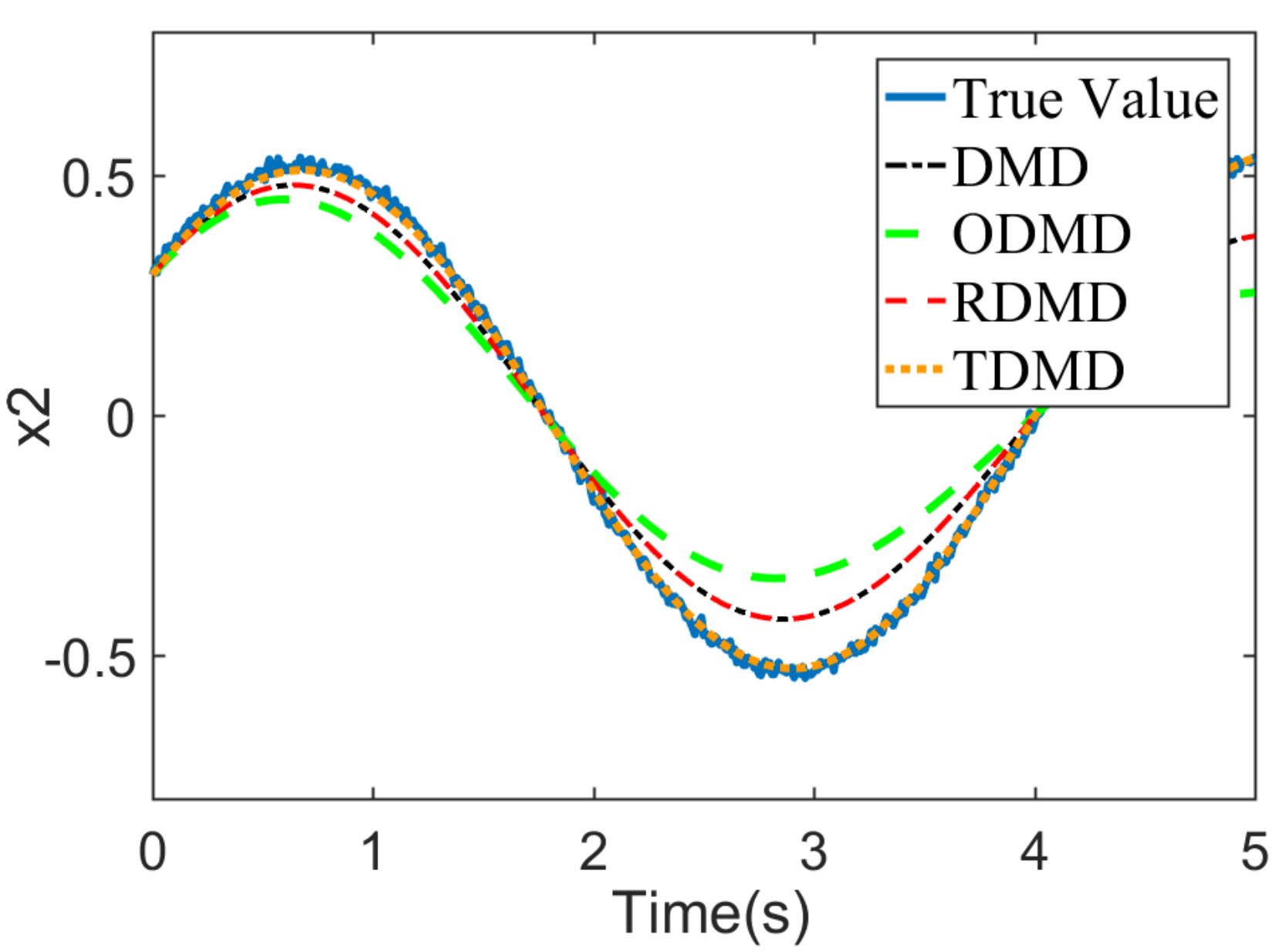}}
\subfloat[]{\includegraphics[width=0.25\textwidth]{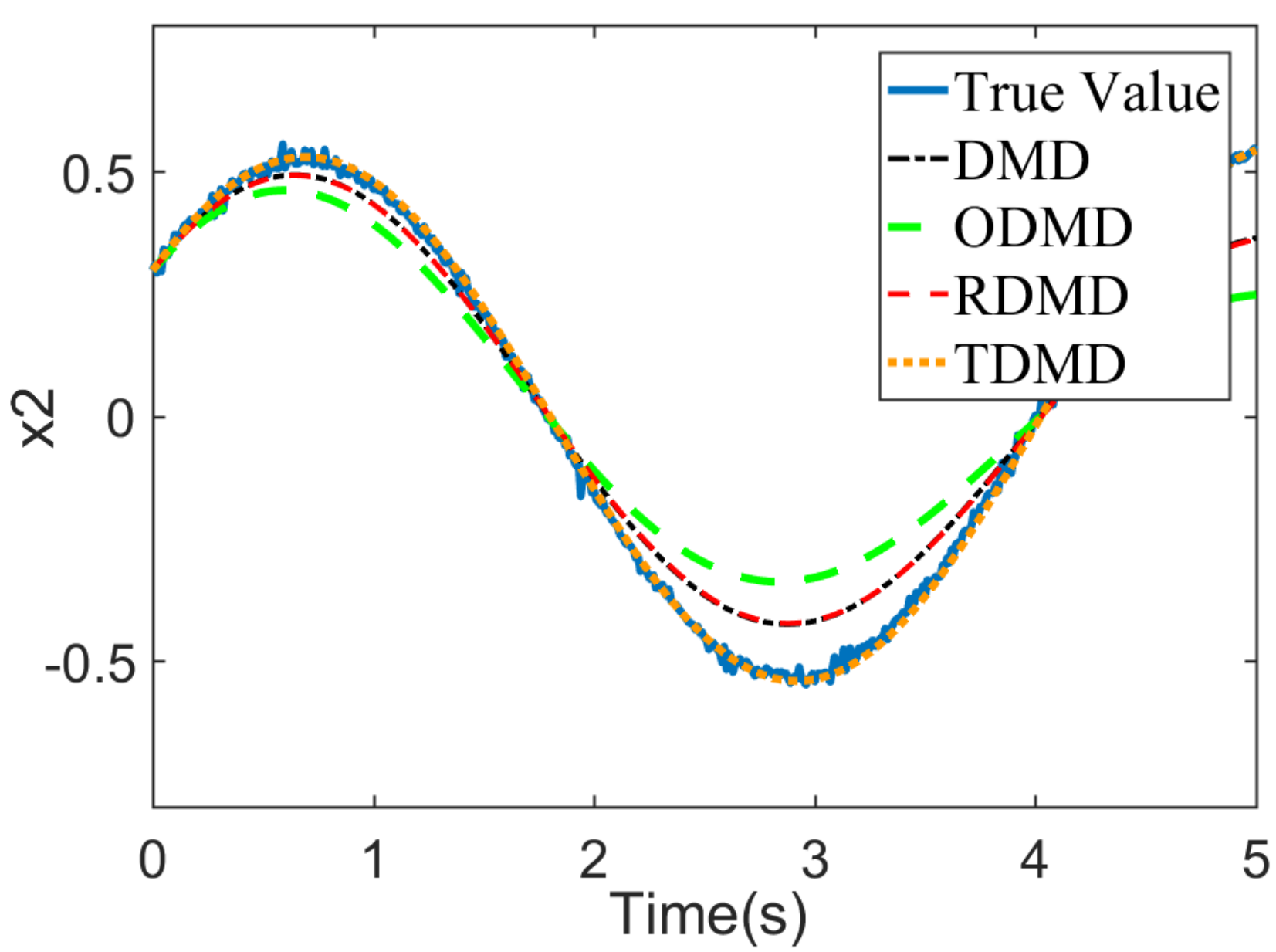}} \\
\subfloat[]{\includegraphics[width=0.25\textwidth]{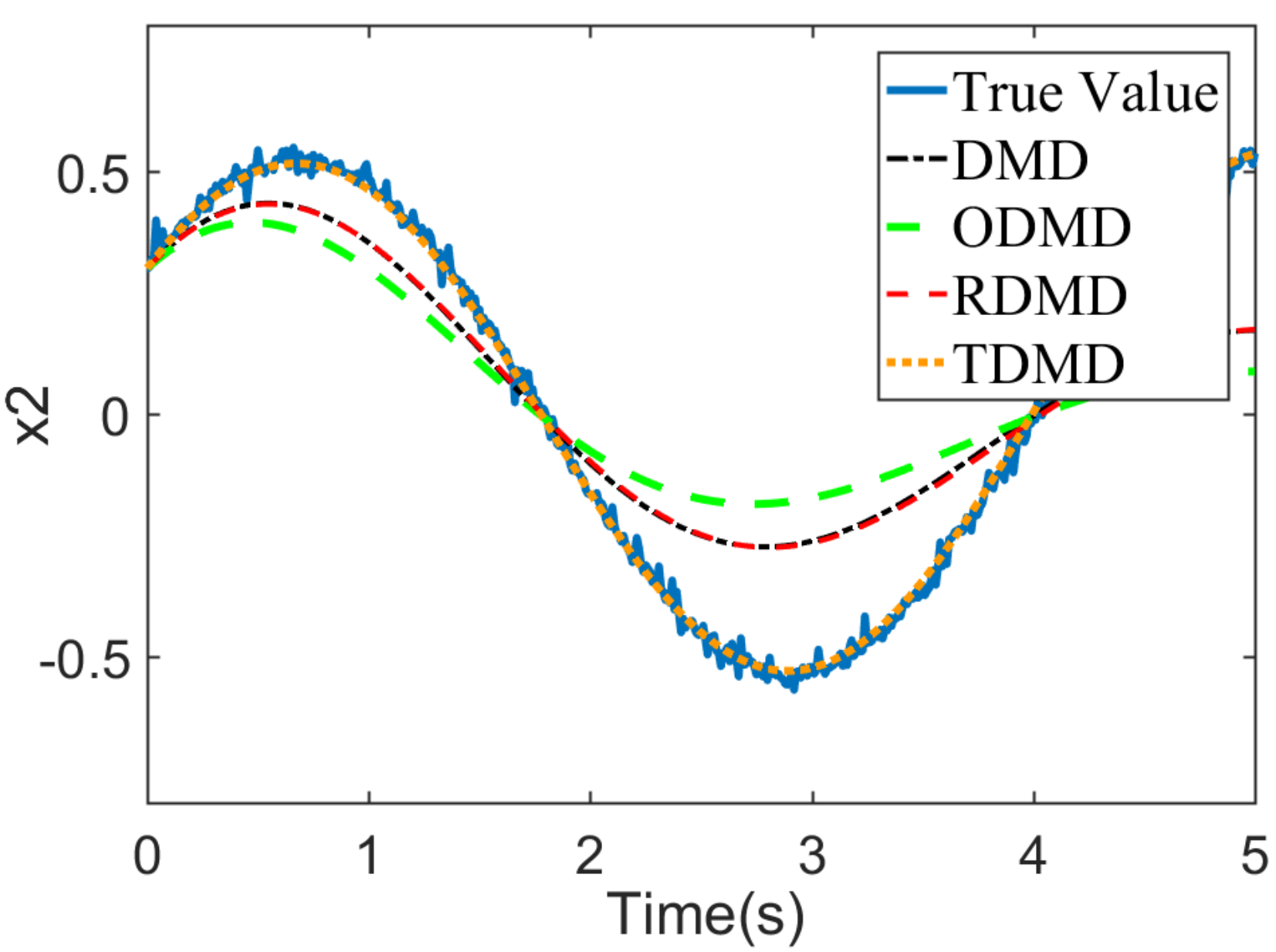}}
\subfloat[]{\includegraphics[width=0.25\textwidth]{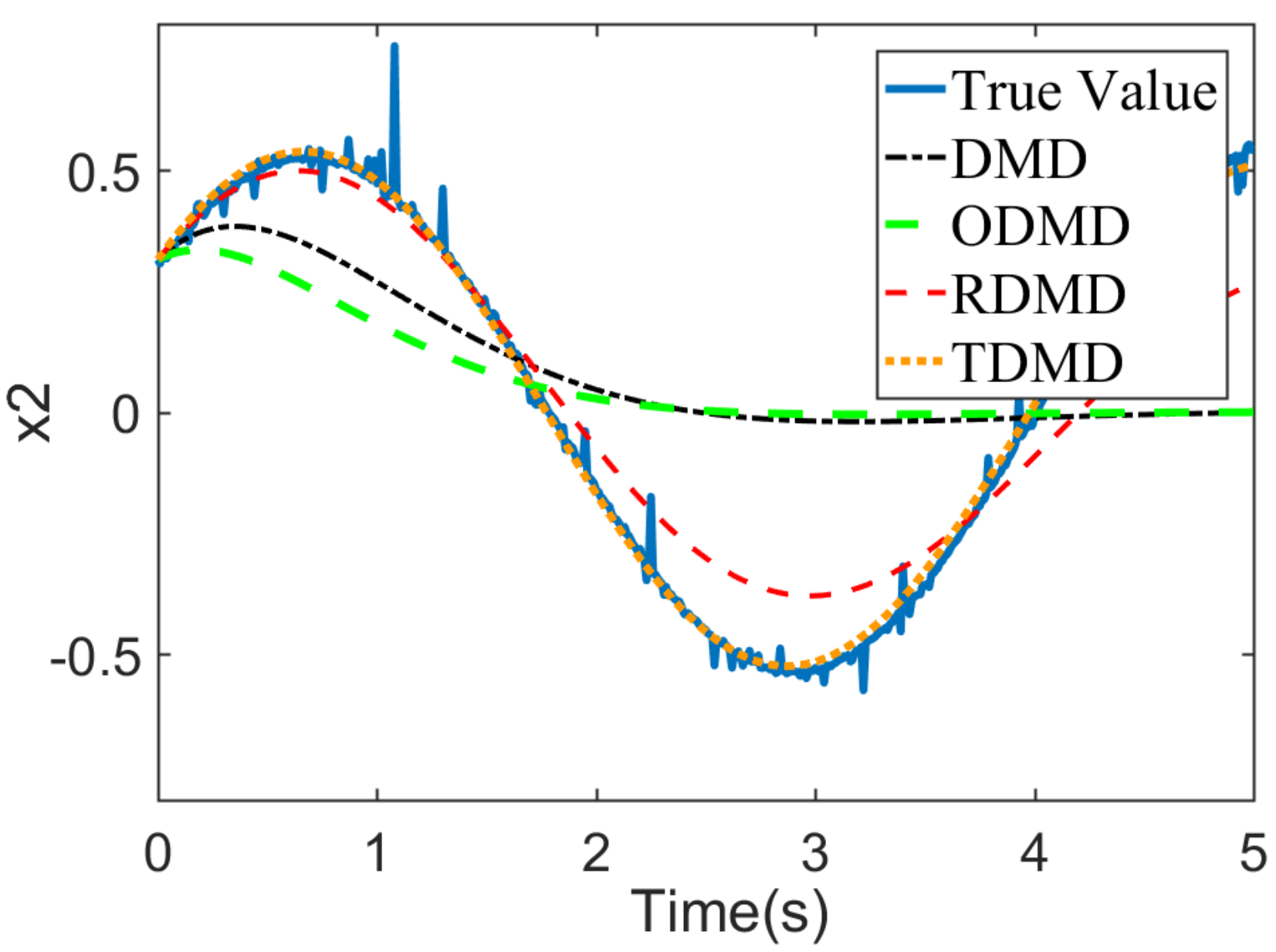}}
\caption{State reconstruction of linear system (\ref{eq.55x}) when the sampled data are contaminated with (a) Gaussian noise with variance of $0.01$; (b) Laplace noise with variance of $0.01$; (c) Student-t noise with $2$ degrees of freedom; (d) Cauchy noise with half-width at half-maximum $\gamma=2$.}
\label{fig.10x}
\end{figure}

\begin{figure}[!t]
\centering

\subfloat[]{\includegraphics[width=0.24\textwidth]{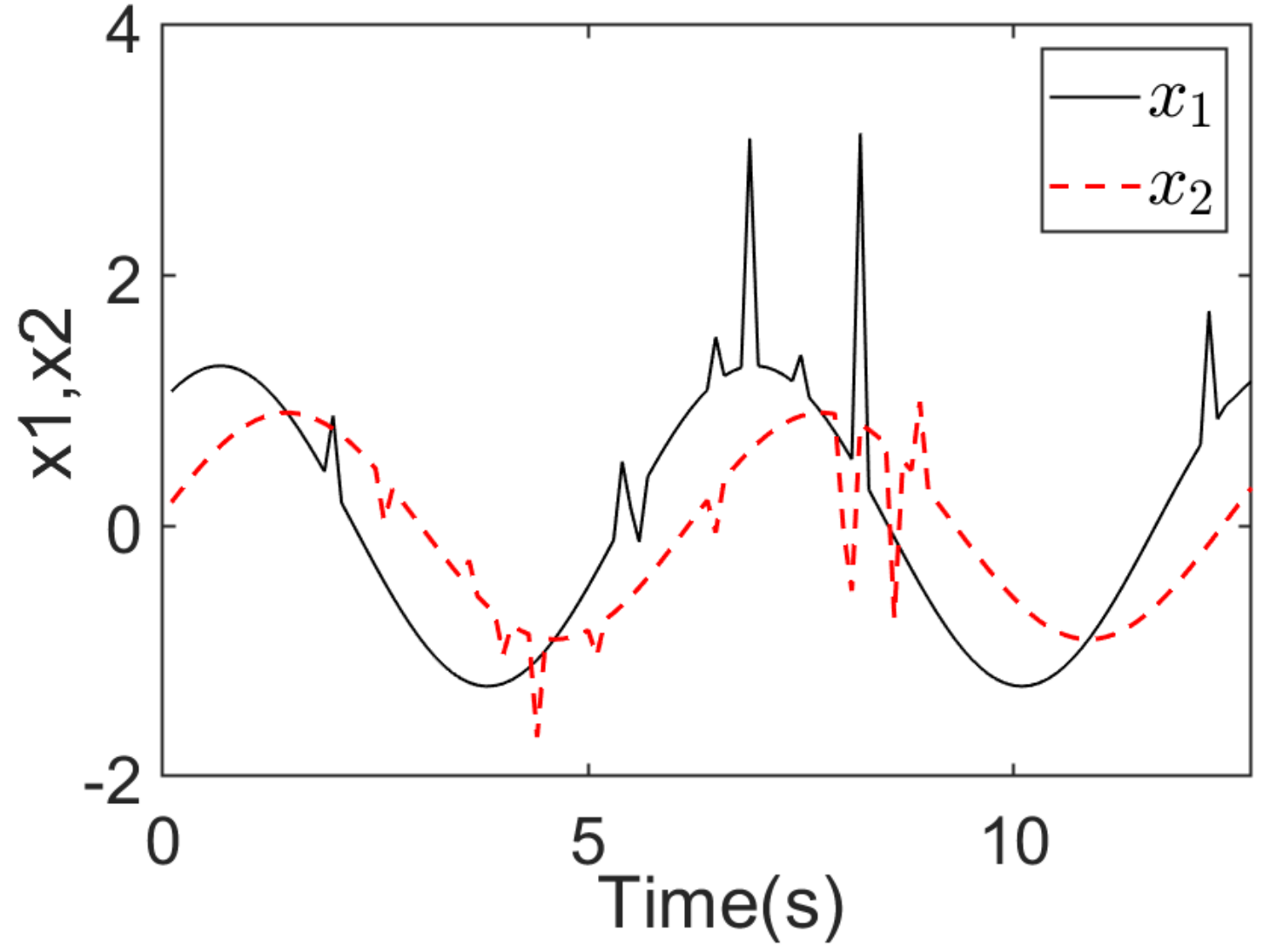}}
\subfloat[]{\includegraphics[width=0.24\textwidth]{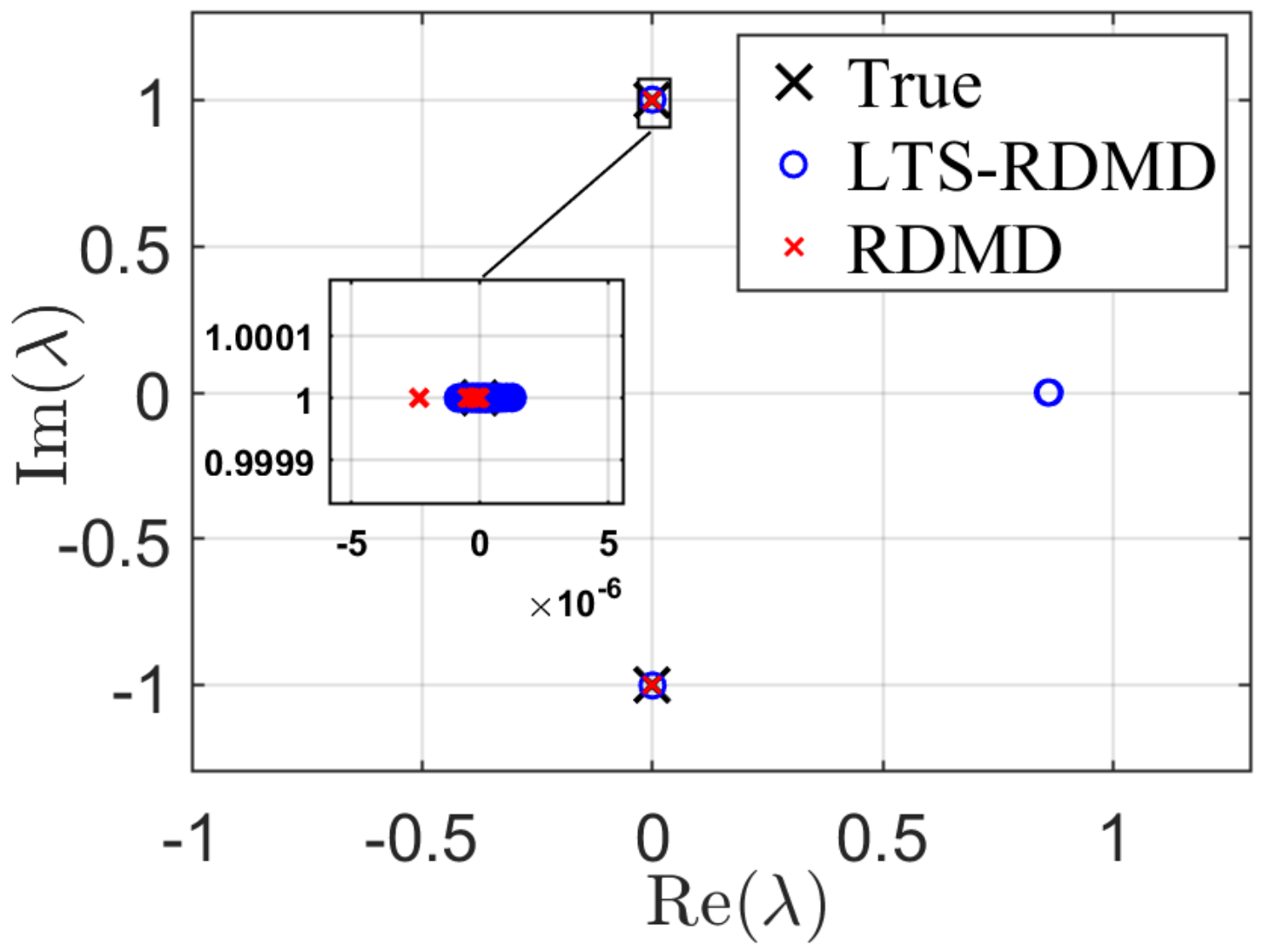}} \\
\subfloat[]{\includegraphics[width=0.24\textwidth]{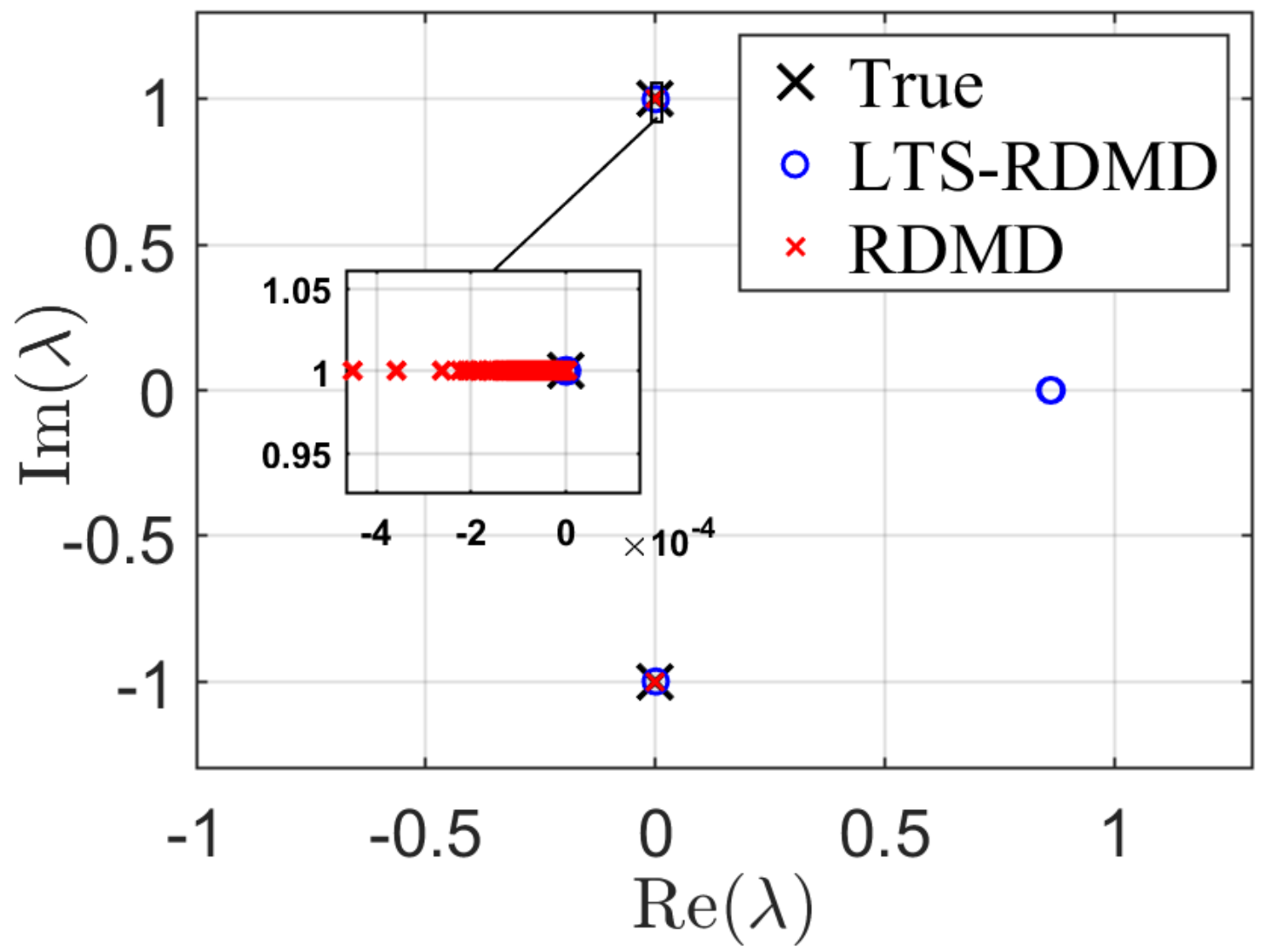}}
\subfloat[]{\includegraphics[width=0.24\textwidth]{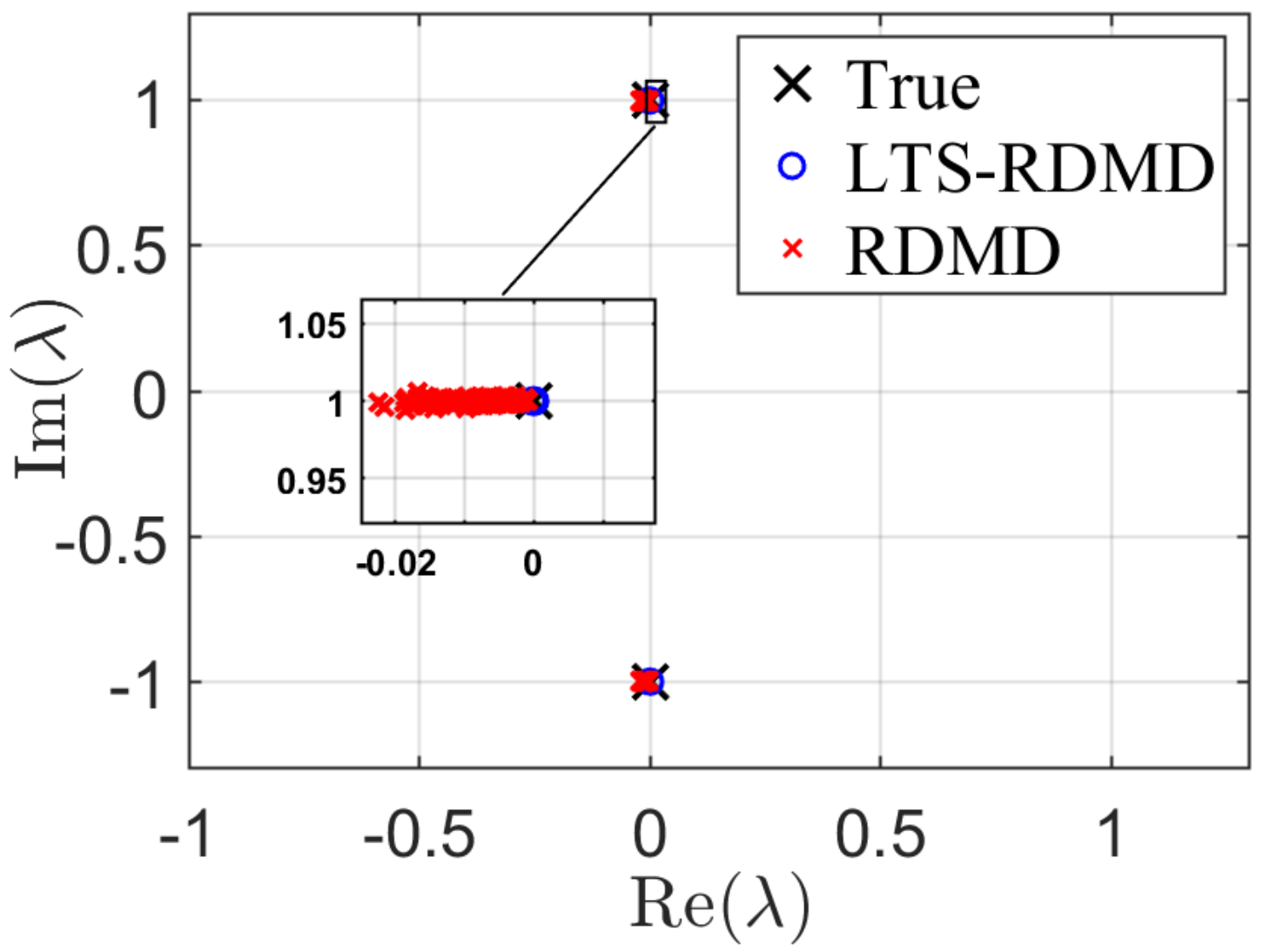}}
\caption{Calculation of eigenvalues by the RDMD and the LTS-RDMD for the simple oscillator (\ref{eq_simple_oscillator}) in the following cases: (a) the measured data with noise and spikes; 
(b) discovered eigenvalues for noise-free data with spike levels ${\mu}=1$, $p=0.05$; (c) discovered eigenvalues with spike levels ${\mu}=1$, $p=0.05$ and Gaussian noise level $\eta=10^{-4}$; (d) discovered eigenvalues with spike levels ${\mu}=1$, $p=0.05$ and Gaussian noise level $\eta=10^{-3}$}
\label{fig.11x}
\end{figure}

\section{Conclusions}
The problem of making the DMD robust to outliers is investigated. By casting the DMD problem in the robust statistics framework, it is solved by using a {\color{black}Schweppe-type Huber generalized maximum-likelihood} estimator. The numerical results demonstrated the effectiveness of the proposed RDMD for a variety of dynamical systems when the sampled data are contaminated with outliers. Further, the proposed RDMD presented satisfactory performance in dealing with non-Gaussian noises. Finally, we noticed that the numerical results are significantly improved by considering the symmetry of the problem; thus, we suggest the robustification of the total least-squares method as a direction for future research.

% \appendices 
% \section{}\label{app.A}
% \section*{Acknowledgments}

\bibliographystyle{IEEEtran}
% \bibliography{lib}
% Generated by IEEEtran.bst, version: 1.14 (2015/08/26)

\end{document}